\documentclass[a4paper, 11pt]{article}
\usepackage[usenames]{xcolor}
\usepackage{xifthen}
\usepackage{algorithmic}
\usepackage{comment}
\usepackage{tcolorbox}
\usepackage{relsize}
\usepackage[font=footnotesize]{caption}
\usepackage[normalem]{ulem}
\usepackage{bm}

\def\fontsettingup{2} 

\usepackage{amsmath, amsthm, amssymb}
\usepackage[T1]{fontenc}
\ifthenelse{\fontsettingup = 1}{ \usepackage{eulerpx, eucal, tgpagella}   }{}
\ifthenelse{\fontsettingup = 2}{ \usepackage{mathpazo, eufrak, tgpagella} }{}

\usepackage{hyperref}
\hypersetup{
  colorlinks=true,
  linkcolor=blue,
  filecolor=magenta,
  urlcolor=cyan,
  citecolor=blue
}

\usepackage[linesnumbered,boxed,ruled,vlined]{algorithm2e}
\usepackage{tikz}
\usepackage{cleveref}
\usepackage{makecell}

\usepackage[a4paper]{geometry}
\newgeometry{
textheight=9in,
textwidth=6.5in,
top=1in,
headheight=14pt,
headsep=25pt,
footskip=30pt
}

\newtheorem{theorem}{Theorem}

\newtheorem*{claim*}{Claim}

\newtheorem{fact}[theorem]{Fact}
\newtheorem{lemma}[theorem]{Lemma}

\theoremstyle{definition}

\newtheorem{definition}[theorem]{Definition}
\newtheorem{remark}[theorem]{Remark}
\newtheorem*{remark*}{Remark}

\ifthenelse{\fontsettingup = 1}{
  \def\*#1{\mathbf{#1}} 
  \def\+#1{\mathcal{#1}} 
  \def\-#1{\mathrm{#1}} 
  \def\^#1{\mathbb{#1}} 
  \def\!#1{\mathfrak{#1}} 
}{}

\ifthenelse{\fontsettingup = 2}{
  \def\*#1{\boldsymbol{#1}} 
  \def\+#1{\mathcal{#1}} 
  \def\-#1{\mathrm{#1}} 
  \def\^#1{\mathbb{#1}} 
  \def\!#1{\mathfrak{#1}} 
}{}

\def\oPr{\mathbf{Pr}}
\renewcommand{\Pr}[2][]{ \ifthenelse{\isempty{#1}}
  {\oPr\left[#2\right]}
  {\oPr_{#1}\left[#2\right]} } 

\def\oE{\mathbb{E}}
\newcommand{\E}[2][]{ \ifthenelse{\isempty{#1}}
  {\oE\left[#2\right]}
  {\oE_{#1}\left[#2\right]} }

\DeclareMathOperator*{\oVar}{\mathbf{Var}}
\newcommand{\Var}[2][]{ \ifthenelse{\isempty{#1}}
  {\oVar\left[#2\right]}
  {\oVar_{#1}\left[#2\right]} }

\def\oEnt{\mathbf{Ent}}
\newcommand{\Ent}[2][]{ \ifthenelse{\isempty{#1}}
  {\oEnt\left[#2\right]}
  {\oEnt_{#1}\left[#2\right]} }

\renewcommand{\epsilon}{\varepsilon}

\newcommand{\norm}[1]{\left\Vert#1\right\Vert}

\newcommand{\lap}{\texttt{Lap}}
\newcommand{\maxwt}{W_{\mathsf{max}}}

\newcommand{\Paren}[1]{\left(#1\right)}

\newcommand{\Erdos}{Erd\H{o}s\xspace}
\newcommand{\Renyi}{R\'enyi\xspace}

\let\epsilon=\varepsilon

\title{Optimal Bounds on Private Graph Approximation}
\date{}
\author{Jingcheng Liu\thanks{State Key Laboratory for Novel Software Technology, Nanjing University, 163 Xianlin Avenue, Nanjing, Jiangsu Province, China. \textnormal{E-mail: \url{liu@nju.edu.cn}, \url{zou.zongrui@smail.nju.edu.cn}}.}
\and 
Jalaj Upadhyay\thanks{Rutgers University, USA. \textnormal{E-mail: \url{jalaj.upadhyay@rutgers.edu}}. Research supported in part by Rutgers Decanal Grant no. 302918.}
\and 
Zongrui Zou\footnotemark[1]
}

\begin{document}

\maketitle

\begin{abstract}
  
   We propose an efficient $\epsilon$-differentially private algorithm, that given a simple {\em weighted} $n$-vertex, $m$-edge graph $G$ with a \emph{maximum unweighted} degree $\Delta(G) \leq n-1$, outputs a synthetic graph which approximates the spectrum with $\widetilde{O}(\min\{\Delta(G), \sqrt{n}\})$ bound on the purely additive error. To the best of our knowledge, this is the first $\epsilon$-differentially private algorithm with a non-trivial additive error for approximating the spectrum of the graph. One of the subroutines of our algorithm also precisely simulates the exponential mechanism over a non-convex set,  which could be of independent interest given the recent interest in sampling from a {\em log-concave distribution} defined over a convex set. As a direct application of our result, we give the first non-trivial bound on approximating all-pairs {\em effective resistances} by a synthetic graph, which also implies approximating {\em hitting/commute time} and {\em cover time} of random walks on the graph. Given the significance of effective resistance in understanding the statistical properties of a graph, we believe our result would have further implications.

   Spectral approximation also allows us to approximate all possible $(S,T)$-cuts, but it incurs an error that depends on the maximum degree, $\Delta(G)$. We further show that using our sampler, we can also output a synthetic graph that approximates the sizes of all $(S,T)$-cuts on $n$ vertices weighted graph $G$ with $m$ edges while preserving $(\epsilon,\delta)$-differential privacy and an additive error of $\widetilde{O}(\sqrt{mn}/\epsilon)$. We also give a matching lower bound (with respect to all the parameters) on the private cut approximation for weighted graphs. This removes the gap of $\sqrt{W_{\mathsf{avg}}}$ in the upper and lower bound in Eli{\'a}{\v{s}}, Kapralov, Kulkarni, and Lee (SODA 2020), where $W_{\mathsf{avg}}$ is the average edge weight. 
\end{abstract}

\thispagestyle{empty}

\clearpage
\pagenumbering{arabic}

\tableofcontents

\newpage

\section{Introduction}

Graph data provides a flexible and intuitive framework to represent and analyze complex relationships, making it ideal for modeling the intricate interactions between individuals. As a widely used example, social networks serve as real-world manifestations of graph data, where individuals or entities are represented as vertices, and the relationships or interactions between them are represented as edges. These data are inherently sensitive and confidential, which makes the publication of such graph data, or releasing information on the relationships between entities, seen as a threat to user's privacy~\cite{hay2009accurate, abawajy2016privacy, sun2019analyzing}.
So, it is imperative to perform such analysis under a robust privacy guarantee. One such privacy guarantee is  {\em differential privacy}  (DP)~\cite{dwork2006calibrating}, which recently has not only spurred extensive research in both theoretical and practical aspects of privacy-preserving techniques but also found application in the real world~\cite{abowd2021geographic, Apple, ding2017collecting, rappor, jiang2021applications, kenthapadi2018pripearl, rogers2020linkedin, wang2022privlbs}.
\begin{definition}[$(\epsilon,\delta)$-differential privacy]\label{def.dp}
    Let $\mathcal{A}:\mathcal{D}\rightarrow \mathcal{R}$ be a randomized algorithm, where $\mathcal{R}$ is the output domain. For fixed $\epsilon >0$ and $\delta\in [0,1)$, we say that $\mathcal{A}$ preserves $(\epsilon,\delta)$-differential privacy if for any measurable set $S\subset \mathcal{R}$ and any pair of neighboring datasets $x,y\in \mathcal{D}$, it holds that 
    $$\mathsf{Pr}[\mathcal{A}(x)\in S] \leq \mathsf{Pr}[\mathcal{A}(y)\in S]\cdot e^{\epsilon} + \delta.$$
    If $\delta = 0$, we also say $\mathcal{A}$ preserves pure differential privacy (denoted by $\epsilon$-DP).
\end{definition}

In this paper, we study the problem of generating a synthetic graph that maintains certain algebraic (i.e., spectrum) and combinatorial properties (i.e., cut function) of the input graph while ensuring differential privacy. We identify an $n$ vertices graph $G=(V,E,w)$ by its weight vector  $w \in \mathbb{R}_{+}^{N}$ that encodes the weight on the edges of the graph. 
Here, $\mathbb{R}_{+}$ denotes the set of non-negative real numbers and $N={n\choose 2}$ is the total number of possible edges. 
Every coordinate of $w$, denoted by $w_e$, corresponds to an edge $e=\{u,v\}$ identified by its endpoints $u,v\in V$. With this notation, the sum of weights of all edges in the graph is simply its $\ell_1$-norm, $\|w\|_1$, the number of edges is $m=\|w\|_0$, and the maximum edge weight is the $\ell_\infty$-norm, $\maxwt = \| w\|_\infty$. For the ease of readers, we present all our notations in Table~\ref{tab:notation}. 

The definition of differential privacy relies on the notion of neighboring datasets. We use the standard notion of differential privacy on graphs known as {\em edge-level differential privacy}: two graphs $G=(V,E,w)$ and $G'=(V,E',w')$ with $w,w' \in \mathbb R_{+}^{N}$ are {\em neighboring} if $\| w - w'\|_0 \leq 1$ and $\| w- w'\|_\infty \leq 1$. That is, $G$ and $G'$ differ in one edge by weight at most {\em one}. 
Equipped with these definitions, we give an overview of our contributions on private graph approximations followed by an overview of our techniques.

\paragraph{Contribution I: Approximating Graph Spectrum}

Algebraic graph theory is a subarea of mathematics that studies how combinatorial properties of graphs are related to algebraic properties of associated matrices (adjacency or Laplacian matrices). For example,  
the graph Laplacian encodes many important graph properties~\cite{biggs1993algebraic, godsil2001algebraic, spielman2007spectral}  and unsurprisingly estimating it accurately has been studied extensively~\cite{batson2012twice, cohen2018approximating, kelner2013simple, peng2014efficient,spielman2011spectral}. The question can be posed as follows: 
given a positively weighted $n$-vertex graph $G$, output a synthetic graph $\widehat{G}$ such that $\|L_G - L_{\widehat{G}}\|_2$ is minimized, where $\| A \|_2$ denote the spectral norm of matrix $A$. Here, $L_G\in \mathbb{R}^{n\times n}$ denotes the Laplacian matrix of $G$.
Let $D_G\in \mathbb{R}^{n\times n}$ be a diagonal matrix with weighted degrees on the diagonal: $D_G[u,u] = \sum_{e=\{u,v\}\in E} w_e$ for every $u\in [n]$, and $A_G$ be the adjacency matrix of $G$, then $L_G := D_G - A_G$ (see also Definition~\ref{defn.laplacian}).

There has been some work on differentially private estimation of graph spectrum with the best-known result ensuring that  $\|L_G - L_{\widehat{G}}\|_2 =\widetilde{O}(\sqrt{n}/\epsilon)$ with $(\epsilon,\delta)$-differential privacy~\cite{dwork2014analyze}. This is meaningful only for dense graphs ($i.e., m = \Omega(n^{3/2})$). 
In contrast, large-scale graphs, such as social networks, are typically sparse, with vertex degrees often negligible compared to the scale of $n$ (see \cite{goswami2021sparsity} and references therein). This raises the following question:
\begin{quote}
  {\em Is there a private algorithm that achieves a better utility (than $O(\sqrt{n})$) on the spectral approximation for sparse graphs, or graphs with a bounded (unweighted) degree, i.e., maximum number of edges incident on any vertex?}
\end{quote}

We answer this question in the affirmative. In particular, we show the following:

\begin{theorem}
[Informal version of Theorem~\ref{t.spectral}]
\label{t.intro_spectrum}
  Let $n\in \mathbb N$. Fix any $\epsilon>0$. There is a polynomial time $\epsilon$-differentially private algorithm such that, on input an $n$-vertex graph $G$ with maximum unweighted degree $\Delta(G)$, it output a synthetic graph $\widehat{G}$ such that with high probability,
  $$\left\|L_G - L_{\widehat{G}}\right\|_2 = O\left(\frac{\Delta(G)\log^2(n)}{\epsilon}\right).$$
\end{theorem}

\begin{table}[t!]
  \centering
  \caption{The comparison of existing differentially private algorithms on spectral approximation.}\label{table.spectral}
  \begin{tabular}{|c|c|c|c|c|}
      
      \hline
      \textbf{Method} &  \makecell[c]{\textbf{Additive error on} \\\textbf{spectral approximation}} &  \makecell[c]{$\delta = 0$\textbf{?}} & \makecell[c]{\textbf{Releasing a}\\\textbf{synthetic graph?}} & \makecell[c]{\textbf{Additive error on}\\\textbf{commute/cover time}} \\
      \hline 
      \makecell[c]{JL mechanism \\ \cite{blocki2012johnson}} & ${{O}}\left( \frac{\sqrt{n}\log^{1.5}(n/\delta)}{\epsilon}\right)$ & No  &  No & $\widetilde{O}\left(\|w\|_1\cdot \frac{1}{\lambda^2(G)}\right)$ \\
      \hline
      \makecell[c]{Analyze Gauss \\\cite{dwork2014analyze}} & ${{O}}\left( \frac{\sqrt{n}\log(n)\log(1/\delta)}{\epsilon}\right)$ & No  & Yes & $\widetilde{O}\left(\|w\|_1\cdot \frac{\sqrt{n}}{\lambda^2(G)}\right)$ \\
      \hline 

      {This paper} & {${O}\left(\frac{\Delta(G)\log^2(n)}{\epsilon}\right)$} &  \textbf{Yes} & \textbf{Yes} & $\widetilde{O}\left(\|w\|_1\cdot \frac{\Delta(G)}{\lambda^2(G)}\right)$\\
      \hline 
  \end{tabular}
  \vspace{-3mm}
\end{table}
\noindent We summarize a comparison of our bounds on spectral approximation and related graph parameters with previous works in Table \ref{table.spectral}\footnote{In contrast with our result, the JL mechanism~\cite{blocki2012johnson} introduces a multiplicative error.}. In the table, $\widetilde{O}(\cdot)$ also hides the $\epsilon$ and $\log(1/\delta)$ term (for $\delta \neq 0$). For spectral approximation, the second column of \Cref{table.spectral} clearly indicates that we achieve better accuracy when $\Delta(G) = o(\sqrt{n})$. Notably, this is the first $\epsilon$-DP algorithm for graph spectral approximation with a non-trivial accuracy guarantee\footnote{One can use the low-rank approximation algorithm~\cite{kapralov2013differentially} to output a spectral approximation of a graph while preserving $\epsilon$-differential privacy. This would result in an additive error $\beta \lambda_1(G)$ as long as $\lambda_n(G) = \Omega(n^4/\epsilon \beta)$, where $\lambda_n(G)$ is the largest eigenvalue of $L_G$.}.

\paragraph{Applications of Estimating Graph Spectrum.}

The spectrum of a graph captures many important graph parameters and statistics. In particular,   
through the connection between random walks and electrical networks, preserving the graph spectrum means that many statistics of random walks are also preserved: this includes hitting time, commute time, and cover time
~\cite{chandra1989electrical,tetali1991random}. These statistics play important roles in the analysis of graph structure, especially in random walk-based algorithms~\cite{von2014hitting}. Further, the all-pairs effective resistances of a graph can also define an $\ell_2$ metric and has been found useful in giving a deterministic constant-factor approximation algorithm for the cover time~\cite{ding2011cover}. 
In Section \ref{s.application}, we present a direct application of Theorem \ref{t.intro_spectrum}, demonstrating how to approximate these critical graph parameters for \textit{connected} graphs. A comparison with existing approaches is also shown in Table \ref{table.spectral}. In the fifth column of Table \ref{table.spectral}, we use $\lambda(G)$ to denote the spectral gap, namely the smallest non-zero eigenvalue of the graph Laplacian matrix $L_G$.  

To achieve these bounds, our synthetic graph approximates all-pairs effective resistances (Definition~\ref{defn.effective_resistance}) within $\widetilde{O}(\Delta(G)/\lambda^2(G))$ error. 
We show in \Cref{s.application} that the sensitivity of one-pair effective resistance is already $\Theta(1/\lambda^2(G))$, so a naive approach would incur an error $\Theta(n^2/\lambda^2(G))$ while ensuring $\epsilon$-DP. In contrast, we show how to release a synthetic graph that approximates all-pairs effective resistances while ensuring $\epsilon$-DP by paying only an extra factor of $\Delta(G)$.

\paragraph{Contribution II: Cut Queries on a Graph.} One of the most prominent problems studied in private graph analysis is preserving cut queries~\cite{arora2019differentially, gupta2012iterative, blocki2012johnson, dwork2014analyze, eliavs2020differentially}: given a graph $G$, the goal is to output a synthetic graph $\widehat G$ that preserves the sizes of all cuts in $G$, i.e.,  minimize 
\[
\alpha = \max_{S,T \subseteq V, \atop S\cap T  = \varnothing}\left\vert \Phi_G(S,T) - \Phi_{\widehat G}(S,T)  \right\vert, \quad \text{where} \quad \Phi_G(S,T) := \sum_{e=\{s,t\}, s \in S, t \in T} w_e
\]
denote the size of $(S,T)$-cut in $G$. When $T=V\backslash S$, we write $\Phi_G(S)$ instead of $\Phi_G(S,V\backslash S).$

Using the Analyze Gauss algorithm~\cite{dwork2014analyze}, one can achieve this goal while incurring an additive error, $\alpha = \widetilde{O}(n^{3/2}/\epsilon)$, where $\widetilde{O}(\cdot)$ hides the $\text{poly}(\log(n))$ factor. As mentioned earlier, this is nontrivial only for dense graphs (i.e., $\|w\|_0 \gg n^{3/2}$). This can be improved using Theorem~\ref{t.intro_spectrum} when $\Delta(G)=o(\sqrt{n})$. Recently, Eli{\'a}{\v{s}} et al.~\cite{eliavs2020differentially} showed the following:  
\begin{itemize}
    \item There is an efficient algorithm that achieves $\alpha = \widetilde{O}\Paren{\sqrt{\|w\|_1n \over \epsilon}}$~\cite[Theorem 1.1]{eliavs2020differentially}.
    
    \item To complement their upper bound, they showed the following lower bound on achievable accuracy:
    \begin{enumerate}
    \item [(i)] for any $(\epsilon,\delta)$-differentially private algorithm, $\alpha = \Omega\Paren{(1-c)\sqrt{\| w\|_0n  \over\epsilon}}$ for an unweighted graph (where $c = {9\delta (e-1) \over e^\epsilon-1}$); and
    \item [(ii)] for any $\epsilon$-differentially private algorithm, $\alpha =\Omega\Paren{\sqrt{\|w\|_0n} \over \epsilon}$ for graphs with  $\|w\|_1=\Theta(\|w\|_0/\epsilon)$. 
    \end{enumerate} 
    \end{itemize}

So, for unweighted graphs, $\alpha = \widetilde \Theta\Paren{\sqrt{\|w\|_0 n/\epsilon}}$ for any $(\epsilon,\delta)$-differentially private algorithm; however, most real-world graphs are weighted~\cite{liu2009privacy, garlaschelli2009weighted, das2010anonymizing, pacheco2020uncovering} and $\|w\|_0$ and $\|w\|_1$ do not coincide (even asymptotically). In fact, a common assumption made in graph theoretic algorithms is that weights are some polynomial $f(n)$. 
In this case, $\|w\|_0$ and $\|w\|_1$  can differ by a factor of $f(n)$. Even from a purely theoretical point of view, this discrepancy between upper and lower bounds for weighted graphs is unsatisfying. This motivates us to study the question:
\begin{quote}
  {\em A weighted graph has many parameters: number of vertices, number of edges, maximum weights on any edge, and sum of weights, etc. Which of these parameter(s) characterizes the accuracy guarantee of differentially private cut approximation?}
\end{quote}

We resolve this question. In particular, we show the correct dependency of $\alpha$ with respect to the number of edges, $\|w\|_0$, and the privacy parameter, $\epsilon$: 

\begin{theorem}
[Informal version of Theorem~\ref{t.optimal_error} and \ref{t.lowerbound_cut}]
\label{t.intro_cut}
  Let $n \in \mathbb N$ and $N = {n \choose 2}$. Fix $\frac{1}{n}\leq \epsilon<\frac{1}{2}$ and $0<\delta<\frac{1}{2}$. There is a $(\epsilon,\delta)$-differentially private algorithm that, given  an $n$ vertices weighted graph $G = (V,E,w)$ with $w\in \mathbb{R}_{+}^{N}$ as input, outputs a graph $\widehat{G}$ such that with probability at least $1-o(1)$, for any disjoint $S,T\subseteq V$,
  $$|\Phi_G(S,T) - \Phi_{\widehat{G}}(S,T)| =  O\left(\frac{\sqrt{\|w\|_0 n}}{\epsilon}\log^3\left(\frac{n}{\delta}\right)\right).$$
  Further, there is no $(\epsilon,\delta)$-differentially private algorithm that, given a (weighted) graph, answers all $(S, T)$-cuts with $\alpha = o\left({(1-c)\over \epsilon}\sqrt{\|w\|_0 n} \right)$, where $c = {9\delta (e-1) \over e^\epsilon-1}$.

\end{theorem}

The above result shows that the correct dependence on the weight vector is in terms of the number of edges (and not the sum of edge weights, i.e., $\|w\|_1$) for any general graph. 
Recall that the lower bound in \cite{eliavs2020differentially} does not explain the error for weighted graphs -- their lower bound is for unweighted Erdos-Renyi graph. Thus, our lower bound does not contradict it but complements them in the weighted case under approximate differential privacy.  More concisely, our bound identifies that the ``hardest weighted graphs'' is the one that has weights $\approx 1/\epsilon$. 
In fact, our lower bound can also be seen as an extension of \cite[Theorem B.1]{eliavs2020differentially} to approximate differential privacy, which showed a similar lower bound of $\Omega\Paren{\sqrt{\|w\|_0 n} \over \epsilon}$ under pure differential privacy.
Combined with our new upper bound which removes the dependency on the weight vector, this closes the gap of $\sqrt{W_{\mathsf{avg}}}$ in \cite{eliavs2020differentially}.

\subsection{Technical Overview}\label{s.technical_overview}
In this section, we introduce the technical ingredients that underpin our results and discuss the key ideas behind our analysis of privacy and utility guarantees. 

Under the definition of edge-level differential privacy, a pair of neighboring graphs could have a difference in any $N:={n \choose 2}$ pair of vertices. Therefore, for the naive randomized response of adding noise to every coordinate to output a synthetic graph, we usually have to perturb every coordinate, which would lead to a large error. On the other hand, if we have the promise that the topology (set of edges) of the graph is fixed, that is, any pair of neighboring graphs are only allowed to differ in one edge in some given edge set, then intuitively protecting the edge-level differential privacy is much easier, and we can approximate the graph with a better accuracy guarantee.  

In our case, we do not know the topology of the graph, so we need to sample the topology while ensuring privacy. One natural approach is to use the exponential mechanism with an appropriate scoring function~\cite{mcsherry2007mechanism}. However, there are two challenges we need to resolve: 
\begin{enumerate}
    \item If $\mathcal{G}$ is the collection of topologies on graphs with $n$ vertices and at most $m$ edges, then the support of the distribution is of size $O((n^2)^m)$, that is, we have to sample from a distribution containing $O((n^2)^m)$ candidates. 
    \item The scoring functions can be hard to compute.  For example, if the goal is to preserve the sizes of all cuts, then for any input graph $G$, a natural scoring function measures the maximum difference in any cut between the input graph and output graph~\cite{blocki2012johnson}. 
    Unfortunately, this scoring function is computationally hard.
    Therefore, it is crucial in our algorithm design to find a computationally tractable scoring function with a good utility guarantee.
    
\end{enumerate}

There are many sampling algorithms that sample from a log-concave distribution (which induces the same distribution as the exponential mechanism for convex score function) defined over a convex set~\cite{applegate1991sampling, bassily2014private, brosse2017sampling, chen2017vaidya, dalalyan2020sampling, dwivedi2018log, frieze1999log, frieze1994sampling, ganesh2020faster, leake2021sampling, lovasz2006fast,   mangoubi2022faster,mangoubi2022sampling, narayanan2017efficient, sachdeva2016mixing}. 
It is unclear if one can use these sampling algorithms to get an error that depends sublinearly on $n$, let alone on the maximum degree as claimed in \Cref{t.intro_spectrum}. On the other hand, the utility bound on the exponential mechanism suggests that if we optimize over a smaller domain, we get a better utility.
We show that this is indeed the case in our setting. In particular, we restrict to the set of graphs that are ``equally sparse''. Introducing such a sparsity constraint leads to a non-convex domain, yet we are able to design an \emph{exact sampler} of graph topologies that achieve better approximations.
To the best of our knowledge, our topology sampler is the first sampler for exponential mechanisms defined over a non-convex set, and therefore, we believe, it can be of independent interest.

\paragraph{Private sampling a topology on a given graph.}
We show that we can resolve all these challenges if we use $\ell_1$-norm as the scoring function. 
Basically, choosing the $\ell_1$-norm reduces the sampling distribution to a specific form of \emph{Gibbs distribution}: a product distribution under a size constraint. To achieve pure differential privacy, we need to design a \emph{perfect sampler} for the distribution, which means sampling exactly from the distribution without any error.
This is a highly non-trivial step for many Gibbs distributions. We show that, for the $\ell_1$ scoring function, we can design a marginal sampler that helps us inductively sample a topology from the distribution.

More specifically, for any weighted graph $G=(V,E,w)$ with $w \in \mathbb R_{+}^N$ and an integer $0\leq k\leq N$, we sample and publish a new edge set of size $k$ according to the distribution $\pi$:
\begin{align}
    \forall S\subseteq [N] \land |S| = k, \quad \pi(S) \propto \prod_{e\in S} \exp(\epsilon\cdot w_e).
    \label{eq:samplingdistribution}
\end{align}

Intuitively, if a topology ``covers'' more weights, then it is more likely to be chosen. After obtaining the topology, we can relax the edge-level differential privacy to the weaker notion of privacy where the topology is already known to the public. Leveraging this key observation, we are ready to present ideas that establish the desired error bounds on the spectral and cut approximation.

\medskip
\noindent \underline{\emph{Implementing the  topology sampler,}  $\texttt{TS}_{\epsilon}$}:  
We discuss how, given a graph $G=(V,E,w)$ and an integer $k$, to efficiently sample an edge set $S\subseteq [N]$ according to $\pi(S)$ defined in \Cref{eq:samplingdistribution} to preserve $\epsilon$-differential privacy. We first associate each possible edge (even the one with zero weight or not in $E$), $e\in [N]$, with a Bernoulli random variable $X_e \in \{0,1\}$ with $\mathsf{Pr}[X_e = 1] = \frac{\exp({\epsilon w_e})}{\exp({\epsilon w_e}) + 1}$. Then, we only need to sample from the product distribution of $N$ Bernoulli random variables $X = (X_1, \cdots, X_N)$ under the sparsity constraint $\|X\|_0 = k$. 

In order to sample from this distribution, we sample each random variable $X_e$ according to its marginal distribution conditioned on the sparsity constraint and values of $X_{i}$ for $1\leq i <e$. In particular, for $X_1$, we sample a value of $x_1\in \{0,1\}$ from its marginal distribution conditioned on $\|X\|_0 = k$. If $x_1 = 1$, we sample $X_2$ from its marginal distribution conditioned on $\|X_{2:N}\|_0 = k-1$, otherwise we sample its marginal distribution conditioned on $\|X_{2:N}\|_0 = k$. Here, $X_{i:j}$ is the shorthand for $(X_i,\cdots, X_j)$. We repeat this procedure until all values in $X$ are settled. 

The key ingredient for the above sampling to succeed is to compute the exact marginal distribution conditioned on the sparsity constraint. We define this procedure next. Let $\texttt{nnz}(\cdot)$ denote the number of non-zero entries and $\mathsf{Pr}_{\otimes}[\cdot]$ denote the product distribution. Then, for each $1\leq i \leq N$ and $-1\leq q\leq k$, we define  
\begin{align}
    \widehat{p}_i^q = 
    \begin{cases}
        \mathsf{Pr}_{\otimes}[\texttt{nnz}(X_{i:N})= q], &\text{ if }q\geq 0\\
        0 &\text{ if }q<0.  
    \end{cases}
    \label{eq.widehatp_iq}
\end{align}

Then, for each $X_e$, its marginal distribution conditioned on that $s (s\leq k)$ edges before it was chosen is exactly 
$$p_e^{mar} = \frac{p_e \cdot \widehat{p}_{e+1}^{k - (s+1)}}{p_e \cdot \widehat{p}_{e+1}^{k - (s+1)} + (1-p_e) \cdot \widehat{p}_{e+1}^{k - s}}.$$
To compute $p_e^{mar}$, we need to know  $\widehat{p}_i^q$. The following recurrence relation gives the requisite tool: 
\begin{equation}
  \begin{aligned}
      \widehat{p}_{i}^q &= \mathsf{Pr}_{\otimes}[\texttt{nnz}(X_{i:N})= q] \\
      & = p_{i}\cdot \mathsf{Pr}_{\otimes}[\texttt{nnz}(X_{i+1:N})= q-1] + (1-p_i) \cdot \mathsf{Pr}_{\otimes}[\texttt{nnz}(X_{i+1:N})= q]\\
      & = p_{i}\cdot \widehat{p}_{i+1}^{q-1} + (1-p_i) \cdot \widehat{p}_{i+1}^{q}.
  \end{aligned}
  \label{eq.recurrence}
\end{equation}

\noindent Therefore, we have an oracle to exactly sample from the distribution $\pi$ in polynomial time. 

\paragraph{On the Spectral Approximation of Graphs.} 
We use the topology sampler, $\texttt{TS}_{\epsilon}$, introduced above to design an algorithm on the spectral approximation for the Laplacian matrices of weighted graphs with edge-level differential privacy. For any graph with $n$ vertices and $m$ weighted edges, the algorithm for outputting a synthetic graph can be summarized in the following two steps:

\begin{enumerate}
  \item Release a perturbed number of edges $\widehat{m}$, and run the topology sampler on $G$ and $\widehat{m}$ to get the new topology $\widehat{E}\subset [N]$.
  \item Use the Laplace mechanism to reweigh edges in $\widehat{E}$.
\end{enumerate}

 Since two neighboring graphs may not have the same number of edges under edge-level privacy, we do not directly run the topology sampler, $\texttt{TS}_{\epsilon}$, on the graph $G=(V,E,w)$. To analyze the utility of spectral approximation, we show that $\texttt{TS}_{\epsilon}$ approximately preserves the following combinatorial structure: if the input graph $G$ has maximum unweighted degree $\Delta$ (which means that each vertex $u \in V$ in the input graph has a maximum of $\Delta$ incident edges), then with high probability, the output graph has a maximum unweighted degree of $O(\Delta\log(n))$ and it is possible to preserve the weighted degree of each vertex. The latter plays a key role in approximating the spectrum of the Laplacian matrix. 

To show that this combinatorial structure on unweighted degree is approximately preserved, we note that, for every $e\in [N] \backslash E$, the topology sampler assigns an equal probability of selection because these edges have zero weight. This allows us to treat the distribution outside $E$ as a variation of random graph model $\mathcal{G}_{n,k}$ (where $0 < k < N$ and $k\in \mathbb{Z}$ is a parameter), which assigns an equal probability to every unweighted graph with $k$ edges. Then, we use the following theorem and basic concentration bounds to show that the maximum degree in $\mathcal{G}_{n,k}$ should not be large:

\begin{theorem}[Frieze and Karoński~\cite{frieze2016introduction}]\label{t.random}
  Fix any $p\in (0,1)$, let $\mathcal{G}_{n,p}$ be the \Erdos-\Renyi random graph model in which each edge exists independently with probability $p$. Let $\mathcal{P}$ be a monotone increasing graph property and $p = k/N$. Then, for large $n$ such that $p = o(1)$, $1/Np = o(1)$ and $\sqrt{p/N} = o(1)$, 
  $$\mathsf{Pr}[\mathcal{G}_{n,k}\in \mathcal{P}] \leq 3 \mathsf{Pr}[\mathcal{G}_{n,p}\in \mathcal{P}].$$
\end{theorem}
\noindent Let $\widehat{E}$ be the edge set of $\widehat{G}$, where $\widehat{G}$ is the output graph. Next, we will show that with high probability, for any edge in $E \cup \widehat{E}$, 
$|w_e - \widehat{w}_e| = O(\log(n)).$ This property comes from the topology sampler and the concentration bound of Laplace noise. Thus, in the synthetic graph, not only the weight of each edge in the original graph is preserved, but the number of new edges in $\widehat{E}$ incident to any vertex $u\in [n]$ is also bounded by $O(\Delta\log(n))$. With all these preparations, we use the result of Bilu and Linial~\cite{bilu2006lifts} (stated as Lemma~\ref{l.spectral_bilu}) to complete the analysis on the utility in spectral approximation.

 \begin{theorem}
 [Bilu and Linial \cite{bilu2006lifts}]
 \label{t.spectral_bilu}
   Let $A\in \mathbb{R}^{n\times n}$ be a symmetric matrix such that the $\ell_1$ norm of each row of $A$ is at most $\ell$ and all diagonal entries of $A$ has absolute value less than $O(\alpha (\log(\ell/\alpha)+1))$. If for any non-zero $u,v\in \{0,1\}^n$ with $u^\top v = 0$, it holds that 
    $$\frac{|u^\top A v|}{\|u\|_2\|v\|_2} \leq \alpha,$$
    then the spectral radius of $A$ is $\|A\|_2=O(\alpha (\log(\ell/\alpha)+1))$.
\end{theorem}

\paragraph{On the Cut Approximation with Optimal Error Rate.} 
Given an $n$-vertex, $m$-edge graph $G =(V,E,w)$ with maximum degree $\Delta(G)$ and weight vector $w \in \mathbb R_{+}^N$, an approximation on the spectrum of Laplacian matrix $L_G$ with $\widetilde{O}(\Delta(G)/\epsilon)$ error implies an $\widetilde{O}(n \Delta(G)/\epsilon)$ error in approximating $\Phi_G(S)$, since we can write $\Phi_G(S) = \mathbf{1}_S^\top L_G \mathbf{1}_S$. However, it does not meet the $\Omega(\sqrt{mn}/\epsilon)$ lower bound. Prior to this result, \cite{eliavs2020differentially} used private mirror descent and give an $\widetilde{O}(\sqrt{n\|w\|_1/\epsilon})$ error with approximate differential privacy. If $W_{\mathsf{max}}$ is the maximum weight, this can be $\widetilde O(\sqrt{nmW_{\mathsf{max}}/\epsilon})$ in the worst case. To entirely remove the dependency on weights, we leverage an auxiliary property of our topology sampler: with high probability, if an edge in $E$ is not chosen into the new edge set $\widehat{E}$, then $w_e$ satisfies that $w_e = O(\log(n)/\epsilon)$. Since there are at most $m$ edges in $E$ with weight $O(\log(n)/\epsilon)$, then if we run the private mirror descent on the sub-graph only containing these edges, it will result in an $\widetilde{O}(\sqrt{mn}/\epsilon)$ error on the cut approximation in this sub-graph.

On the other hand, for the edges chosen into the new edge set, since the (perturbed) topology is already public, using the concentration of Laplace noise is enough to get the desired bound. Combining the private output of two sub-graphs, we obtain a synthetic graph that approximates the cut of the original graph with an $\widetilde{O}(\sqrt{mn}/\epsilon)$ error, which matches the $(\epsilon,\delta)$-differentially private lower bound in Section \ref{s.lower_bound}, if we omit the logarithm term.

\section{Preliminaries}

Here, we give a brief exposition of differential privacy and spectral graph theory to the level required to understand the algorithms and their analysis.
\subsection{Background on Differential Privacy}
Differential privacy, proposed by \cite{dwork2006calibrating}, is a widely accepted notion of privacy. Here, we formally define differential privacy. Let $\mathcal{D}$ be some domain of datasets. The definition of differential privacy has been shown in Definition \ref{def.dp}. A key feature of differential privacy algorithms is that they preserve privacy under post-processing. That is to say, without any auxiliary information about the dataset, any analyst cannot compute a function that makes the output less private. 
\begin{lemma}[Post processing~\cite{dwork2014algorithmic}]\label{l.post_processing}
  Let $\mathcal{A}:\mathcal{D}\rightarrow \mathcal{R}$ be a $(\epsilon,\delta)$-differentially private algorithm. Let $f:\mathcal{R}\rightarrow \mathcal{R}'$ be any function, then $f\circ \mathcal{A}$ is also $(\epsilon,\delta)$-differentially private.
\end{lemma}

Sometimes we need to repeatedly use differentially private mechanisms on the same dataset, and obtain a series of outputs.

\begin{lemma}[Adaptive composition~\cite{dwork2006calibrating}]\label{l.adaptive_composition}
  Suppose $\mathcal{M}_1(x):\mathcal{D} \rightarrow \mathcal{R}_1$ is $(\epsilon_1,\delta_1)$-differentially private and $\mathcal{M}_2(x,y):\mathcal{D} \times \mathcal{R}_1\rightarrow \mathcal{R}_2$ is $(\epsilon_2,\delta_2)$-differentially private with respect to $x$ for any fixed $y\in \mathcal{R}_1$, then the composition
  $x \Rightarrow (\mathcal{M}_1(x), \mathcal{M}_2(x,\mathcal{M}_1(x)))$
  is $(\epsilon_1 + \epsilon_2, \delta_1 + \delta_2)$-differentially private.
\end{lemma}

\begin{lemma}
    [Advanced composition lemma, \cite{dwork2010boosting}] 
    \label{l.adv_composition}
    For parameters $\epsilon>0$ and $\delta,\delta'\in [0,1]$, the composition of $k$ $(\epsilon,\delta)$ differentially private algorithms is a $(\epsilon', k\delta+\delta')$ differentially private algorithm, where 
    $\epsilon' = \sqrt{2k\log(1/\delta')} \cdot \epsilon + k\epsilon (e^\epsilon - 1).$
\end{lemma}

\noindent Now, we introduce basic mechanisms that preserve differential privacy, which are ingredients that build our algorithm. First, we define the sensitivity of query functions.

\begin{definition}
  [$\ell_p$-sensitivity] Let $f:\mathcal{D}\rightarrow \mathbb{R}^k$ be a query function on datasets. The sensitivity of $f$ (with respect to $\mathcal{X}$) is 
  $\mathsf{sens}_p (f) = \max_{x,y\in \mathcal{D} \atop x\sim y} \|f(x) - f(y)\|_p$.
\end{definition}

\begin{lemma}[Laplace mechanism]\label{l.laplace}
  Suppose $f:\mathcal{D}\rightarrow \mathbb{R}^k$ is a query function with $\ell_1$ sensitivity $\mathsf{sens}_1(f)\leq \mathsf{sens}$. Then the mechanism
  $\mathcal{M}(D) = f(D) + (Z_1,\cdots,Z_k)^\top$
  where $Z_1,\cdots, Z_k$ are i.i.d random variables drawn from $\texttt{Lap}\left({\mathsf{sens}\over\epsilon}\right)$. Given $b>0$,  $\texttt{Lap}\left({b} \right)$ is the Laplace distribution with density $$\texttt{Lap}(x;b) := \frac{1}{2b} \exp\left(-\frac{|x|}{b}\right).$$
\end{lemma}

\subsection{Background on Spectral Graph Theory}

A weighted graph $G = (V,E,w)$ with non-negative edges can be equivalently represented by its weight vector $w\in \mathbb{R}_{+}^N$ after we identify every edge with a number in $[N]$. Here we write $N = {n \choose 2}$, and let $w_e$ ($e\in [N]$) be the edge weight of the edge $e = \{u,v\}$. If some pair of vertices have no edge, the weight is $0$. Let $A_G\in \mathbb{R}^n\times \mathbb{R}^n$ be the symmetric adjacency matrix of $G$. That is,
$A_G[u,v] = A_G[v,u] = w_{\{u,v\}}$, where $u,v\in V$ are vertices, and $w_{\{u,v\}}$ is the weight of edge $uv$. (If $u$ is not adjacent with $v$, then $w_{\{u,v\}} = 0$.) Here, we define the edge adjacency matrix:
\begin{definition}
  [Edge adjacency matrix] Let $G$ be an undirected graph of $n$ vertices and $m$ edges. Consider an arbitrary orientation of edges, then $E_G\in \mathbb{R}^{m\times n}$ is the edge adjacency matrix where 
  $$E_G[e,v] = \left\{
      \begin{aligned}
          &+\sqrt{w_e}, &\text{if } v \text{ is } e\text{'s head,} \\
          &-\sqrt{w_e}, &\text{if } v \text{ is } e\text{'s tail,} \\
          &0, &\text{otherwise.}
      \end{aligned}
  \right.$$ 
  We drop the subscript $G$ when it is clear from the context. 
\end{definition}
\noindent An important object of interest in graph theory is the Laplacian of a graph:
\begin{definition}
  [Laplacian matrix] 
  \label{defn.laplacian}
  For an undirected graph $G$ with $n$ vertices and $m$ edges, the Laplacian matrix $L_G\in \mathbb{R}^{n\times n}$ of $G$ is $$L_G = E_G^\top E_G.$$
\end{definition}

\noindent Equivalently, one can verify that $L_G = D_G - A_G$, where $D_G\in \mathbb{R}^{n\times n}$ is a diagonal matrix with $D_G[u,u] = \sum_{e=\{u,v\}\in E} w_e$ for every $u\in [n]$. Also, we note that for any graph $G$, $L_G \succeq 0$, where $A \succeq$ denotes that $A$ is a positive semi-definite (PSD) matrix. One can easily verify that $L_G \bm{1}_n = 0$, where $\bm{1}_n\in \mathbb{R}^n$ is the all one vector. Let $0 = \lambda_1(G)\leq \lambda_2(G)\leq\cdots\leq \lambda_n(G)$ be the non-negative eigenvalues of $L_G$. For any vector $x$ in $\mathbb{R}^n$, the quadratic form of $L_G$ is $x^\top L_G x  \geq 0$. In particular, one can verify that 
$$x^\top L_G x = \sum_{e=\{u,v\}\in E} w_{e} (x(u) - x(v))^2.$$
For any connected graph $G$, we use $L_G^{\dagger}$ to denote the {\em Moore-Penrose pseudoinverse} of Laplacian $L_G$. We note that $L_G^{\dagger}$ is also a PSD matrix.

One of our main tasks is to preserve the cut size. For any $S,T\subseteq V$, we simply write $\Phi_G(S)$ be the size of $(S,V\backslash{S})$-cut for a graph $G$, and $\Phi_G(S,T)$ be the size of $(S,T)$-cut. For any vertex $v\in V$, we use $\mathbf{1}_v\in \{0,1\}^n$ to denote the column vector with $1$ in the $v$-th coordinate and $0$ anywhere else. Also, for a subset of vertices $S$ we use $\mathbf{1}_S\in \{0,1\}^n$ to denote the identity vector of $S$. It's easy to verify that 
$$\Phi_G(S) = \mathbf{1}_S^\top L_G  \mathbf{1}_S.$$
Therefore, if we obtain an approximation for the quadratic form of Laplacian matrix $L_G$, it's equivalently to obtain an approximation for the sizes of all $(S,V\backslash {S})$-cuts. Also note that if an algorithm preserves all $(S,V\backslash{S})$-cuts with a pure additive error, then it also preserves all $(S,T)$-cuts with a constant factor since,  for any disjoint $S,T\subseteq V$, $$\Phi_G(S,T) = \Phi_G(S,\bar{S}) + \Phi_G(T,\bar{T}) - 2\Phi_G(S\cup T,\overline{S\cup T}).$$
We also introduce the following lemma to bound the sensitivity of $\lambda_2(G)$:
\begin{lemma}\label{l.sens_of_eigengap}
    Let $G$ and $G'$ be two graphs, where $G'$ is obtained from $G$ by adding one edge joining two distinct vertices of $G$ (or increasing edge weight by $1$). Then 
    $$\lambda_2(G) \leq \lambda_2(G') \leq \lambda_2(G) + 2.$$
\end{lemma}

\section{Implementation of the Marginal Sampler and Topology Sampler}\label{s.sampler}

Our algorithm for both spectral and cut approximation uses a sampler, which we call {\em topology sampler}, that samples from a specific Gibbs distribution. Therefore, we first present the implementation of the topology sampler $\texttt{TS}_{\epsilon}$. Let $N = {n\choose 2}$. For any fixed $\epsilon>0$, given $0\leq k \leq N$ and $G = (V,E,w)$, $\texttt{TS}_{\epsilon}$ samples an edge set $S\subseteq [N]$ with size $k$ according to the distribution
$$\pi(S) \propto \prod_{e\in S} \exp(\epsilon\cdot w_e).$$
Note that $\pi$ can also be considered as the distribution defined by the exponential mechanism (with $\ell_1$ norm as its scoring function) over some non-convex set. To efficiently sample an edge set from $\pi$, we consider a more general tool to sample from a product distribution with sparsity constraint, which can be of independent interest.

Formally speaking, we consider the task where given $N\in \mathbb{N}$, let $\mathcal{P}:(0,1)^N\rightarrow \{0,1\}^N$ be a sampling oracle that takes $p_1,p_2,\cdots, p_N$ as input and outputs a configuration $\sigma$ of random variable $X = (X_1, X_2, \cdots, X_N)$ under the sparsity constraint. Here, $X_i$ is a Bernoulli random variable such that $\mathsf{Pr}[X_i] = p_i$ for $i\in[N]$. In other words, we aim to sample from $X|_{\|X\|_0 = k}$. 

Let 
$\mathsf{Pr}_{\otimes}(~\cdot~)$ be the product distribution and $\texttt{nnz}(X)$ be the number of non-zero elements in $X$, namely $\|X\|_0$. Our goal is to sample a configuration $\sigma\in \Omega_k$ with probability exactly $$\mathsf{Pr}_{\otimes}(X = \sigma | \texttt{nnz}(X) = k), \quad \text{where} \quad \Omega_k = \left\{X\in \{0,1\}^N \text{ such that } \texttt{nnz}(X)= k \right\}.$$ 

To do this, we compute the marginal distribution of each coin, $X_i$ for $1 \leq i \leq N$, conditional on $\texttt{nnz}(X)= k$. 
The idea is as follows: We first decide the marginal distribution of the first coin $X_1$, and sample a value in $x_1\in \{0,1\}$ from the marginal distribution. Then, we decide the marginal distribution of $X_2$, conditioned on $X_1 = x_1$ and $\texttt{nnz}(X)= k$. We repeat this procedure until all values of $\{X_i\}_{i\in[N]}$ are settled. Algorithm \ref{alg.p1.intro} formalizes this idea. In the algorithm, we use some notations that we have defined in Section \ref{s.technical_overview}. 
With Algorithm \ref{alg.p1.intro}, we show the following theorem:
\begin{theorem}\label{t.sample_oracle.intro}
    For any $k\in \mathbb{N}$ and $k\leq N$, there exists a sampling oracle $\mathcal{P}_k$ such that for any $\sigma\in \{0,1\}^N$, $\mathsf{Pr}_{\mathcal{P}_k}(\sigma) = \mathsf{Pr}_{\otimes}(X = \sigma | \texttt{nnz}(X) = k)$. Also, $\mathcal{P}_k$ terminates in time ${O}(Nk)$.
\end{theorem}

    \begin{algorithm}[t]
      \caption{{Sampler} $\mathcal{P}_k$}
      \label{alg.p1.intro}
      \KwIn{$N$ parameters $(p_i)_{i\in [N]}$ associated with $N$ coins.}
      \KwOut{A configuration $\sigma$ of $N$ coins such that only $k$ coins are set to be $1$.}
        \For{$1\leq i\leq N$}{
        \For{$q>N-i+1$}{
        $\widehat{p}_i^q = 0$\;
        }}
        Set $\widehat{P}_{N}^0 \leftarrow 1-p_N$, $\widehat{P}_{N}^1 \leftarrow p_N$\;
        Recursively compute $\widehat{p}_i^q$ for remaining $q$ and $i$ by $ \widehat{p}_i^q = p_{i}\cdot \widehat{p}_{i+1}^{q-1} + (1-p_i) \cdot \widehat{p}_{i+1}^{q}$\;
        Let $k' \leftarrow k$\;
        \For{$1\leq i\leq N-1$}{
            Compute $$p_i^{mar} \leftarrow \frac{p_i\cdot \widehat{p}_{i+1}^{(k'-1)}}{ p_i\cdot \widehat{p}_{i+1}^{(k'-1)} + (1-p_i)\cdot\widehat{p}_{i+1}^{(k')} },$$
            
            \textbf{Sample} $x_i\in \{0,1\}$ with $\mathsf{Pr}[x_i = 1] = p_i^{mar}$\; 
            {\bf Set} $k' \gets k'-1$ if $x_i=1$\;
        }
        \textbf{Set} 
        \[
        x_N = \begin{cases}
            0, & \text{if }k'=0; \\
            1, & \text{otherwise}.
        \end{cases}
        \]

        \Return{$\sigma = (x_1, x_2, \cdots, x_N)$}.
    \end{algorithm}

    \begin{proof}
        We first argue the correctness of the sampling procedure. 
        Recall that for each $1\leq i \leq N$ and $-1\leq q\leq m$, we define  
\begin{align}
    \widehat{p}_i^q = 
    \begin{cases}
        \mathsf{Pr}_{\otimes}(\texttt{nnz}(X_{i:N})= q), &\text{ if }q\geq 0\\
        0 &\text{ if }q<0.  
    \end{cases}
    \label{eq.widehatp_iq_in_proof}
\end{align}
    First, we claim that during the execution of Algorithm \ref{alg.p1.intro}, it always holds that $k'\geq 0$. At the start, it is set to be $k$ (Step 1).  Indeed, if $x_i = 1$ results in $k' = 0$, then by the definition, $\widehat{p}_{i+1}^{k'-1}=0$ for any $1\leq i \leq N-1$. This implies that $p_{i+1}^{mar}=0$ and $x_{i+1} \neq 1$. Further, since $p_i<1$ for any $i\in[N]$, both $(1-p_i)$ and $\widehat{p}_{i+1}^{0}$ are strictly positive. Therefore, $k'$ will no longer decrease. Moreover, for $1\leq i\leq N-1$, we have     $$p_i\cdot\widehat{p}_{i+1}^{(k' - 1)}+ (1-p_i)\cdot \widehat{p}_{i+1}^{(k')} >0,$$
    so that each $p^{mar}_i$ is well-defined. Since $k'$ records the number of non-zero entries in $({x}_{1}, \cdots, {x}_{i-1})$, it is easy to verify that Algorithm \ref{alg.p1.intro} does output an entry in $\Omega_k$.
    
        Suppose the oracle outputs $\bm{x} = (x_1, x_2, \cdots, x_N)\in \Omega_k$. Let $k'$ be the number of non-zero entries in $\bm{x}_{1:i-1} = (x_1, \cdots, x_{i-1})$. Then, by line 3 of Algorithm \ref{alg.p1.intro}, we sample each $X_i$ for $i<N$ with probability
        \begin{equation*}
            \begin{aligned}
                \mathsf{Pr}(X_i = 1) = p_i^{mar} &= {p_i\cdot \widehat{p}_{i+1}^{(k'-1)}\over (p_i\cdot \widehat{p}_{i+1}^{(k'-1)} + (1-p_i)\cdot\widehat{p}_{i+1}^{(k')}) }\\
                & = \frac{p_i \cdot \mathsf{Pr}_{\otimes}(\texttt{nnz}(X_{i+1:N})= k'-1)}{p_i\cdot \mathsf{Pr}_{\otimes}(\texttt{nnz}(X_{i+1:N})= k'-1) + (1-p_i)\cdot \mathsf{Pr}_{\otimes}(\texttt{nnz}(X_{i+1:N})= k')}\\
                & = \frac{\mathsf{Pr}_{\otimes}(X_i = 1 \land \texttt{nnz}(X_{i:N}) = k' )}{\mathsf{Pr}_{\otimes}(X_i = 1 \land \texttt{nnz}(X_{i:N}) = k' ) + \mathsf{Pr}_{\otimes}(X_i = 0 \land \texttt{nnz}(X_{i:N}) = k' )}\\
                & = \mathsf{Pr}_{\otimes} (X_i = 1| \texttt{nnz}(X_{i:N})= k').
            \end{aligned}
        \end{equation*}
        We can equivalently write $\mathsf{Pr}(X_i = 1)$ {(for all $1 \leq i < N$)} by 
        \begin{equation}\label{equation1.intro}
        \mathsf{Pr}(X_i = 1) = \mathsf{Pr}_{\otimes}(X_i = 1|\texttt{nnz}(X)= k \land X_{1:i-1} = \bm{x}_{1:i-1})
        \end{equation}
        and \Cref{equation1.intro} also holds for $X_N$. 
        Thus, for any $\bm{x}\in \Omega_k$,
        \begin{equation*}
            \begin{aligned}
                \mathsf{Pr}_{\mathcal{P}_k}(X = \bm{x}) &= \prod_{i = 1}^{d}\mathsf{Pr}_{\mathcal{P}_k}(X_i = x_i | X_{1:i-1} = \bm{x}_{1:i-1}) \\
                &= \prod_{i = 1}^{d} \mathsf{Pr}_{\otimes} (X_i = x_i| X_{1:i-1} = \bm{x}_{1:i-1} \land \texttt{nnz}(X)= k ),
            \end{aligned}
        \end{equation*}
        which is the same as $\mathsf{Pr}_{\otimes}(X = \bm{x}|\texttt{nnz}(X)= k)$. 
        This completes the correctness proof.
        
        Next, we analyze the time complexity of Algorithm \ref{alg.p1.intro}. A main barrier is how to compute each $\widehat{p}_i^q = \mathsf{Pr}_{\otimes}(\texttt{nnz}(X_{i:N})= q)$ for $q\geq 0$. Generally speaking, directly computing $\widehat{p}_i^q$ is hard. 
        However, using the recurrence relation in \Cref{eq.recurrence}, we can easily compute $\widehat p_i^q$ once we have all the terms, i.e., $p_i,\widehat p_{i+1}^{q-1},$ and $\widehat p_{i+1}^q$. Now, computing $\widehat{p}_{i}^0$ and $\widehat{p}_{N}^q$ is simple for $1\leq i \leq N, 0\leq q\leq k$. Therefore, if we compute each $\widehat{p}_{i}^q $ backwards, we have the value of $\widehat{p}_{i+1}^{q-1} + (1-p_i)$ and $(1-p_i) \cdot \widehat{p}_{i+1}^{q}$ before we compute $ \widehat{p}_{i}^{q}$. Note that, the base case is $\widehat{p}_N^*$, where $\widehat{p}_N^0 = 1-p_N$ and $\widehat{p}_N^1 = p_N$. Then, we can compute all $\widehat{p}_i^q$'s using dynamic programming in time ${O}(Nk)$.  Computing the marginal of every edge in $[N]$ then takes ${O}(N)$ time. Therefore, the time complexity is ${O}(Nk)$.
    \end{proof}

\paragraph{Implementation of Topology Sampler}
Equipped with the marginal sampler described above, we can design an efficient sampler for the distribution defined in \Cref{eq:samplingdistribution}. 

\begin{theorem}\label{t.topology_sampler}
  Fix any $\epsilon>0$, there is a topology sampler $\texttt{TS}_{\epsilon}:\mathbb{N} \times \mathbb{R}_{+}^N  \rightarrow 2^{[N]}$ such that for any given $k\in \mathbb{N}$ ($k\leq N$) and undirected weighted graph with $n$ vertices, it outputs an edge set of size $k$ according to distribution defined in \Cref{eq:samplingdistribution}. Further, $\texttt{TS}_{\epsilon}$ runs in time $O(kn^2)$.
\end{theorem}
\begin{proof}
Let $\epsilon>0$ be a fixed parameter. For any edge $e\in [N]$, we associate a Bernoulli random variable $X_e$ such that 
  $$p_e = \mathsf{Pr}[X_e = 1] = \frac{\exp({\epsilon\cdot w_e})}{1+\exp({\epsilon\cdot w_e})}.$$
  We run the oracle $\mathcal{P}_k$, described in Algorithm~\ref{alg.p1.intro}, on $\{p_e\}_{e\in [N]}$. By Theorem \ref{t.sample_oracle.intro}, for any $S \in \{0,1\}^N$ with $\|S\|_0 = k$,

  \begin{equation*}
    \begin{aligned}
      \mathsf{Pr}[X = S] &\propto \prod_{e\in S} \frac{\exp({\epsilon\cdot w_e})}{1+\exp({\epsilon\cdot w_e})} \cdot \prod_{e\in [N]\backslash S} \frac{1}{1+\exp({\epsilon\cdot w_e})} \propto \prod_{e\in S} \exp(\epsilon\cdot w_e),
    \end{aligned}
  \end{equation*}
  which is the distribution defined in \Cref{eq:samplingdistribution}. Also, the run time of $\texttt{TS}_\epsilon$ directly follows from the run time of $\mathcal{P}_k$ in Theorem \ref{t.sample_oracle.intro}.
\end{proof}

\section{Differentially Private Spectral Approximation}

In this section, we describe our algorithm on spectral approximation for undirected graphs with pure differential privacy and give its analysis. 

\subsection{The Algorithm}\label{s.41}

Let $N = {n\choose 2}$, where $n \in \mathbb N$ is the number of vertices. For any $1\leq k\leq N$, we define ${[N]\choose k}$ be the collection of all subsets of $[N]$ with cardinality $k$, where $[N]=\{1,2,\cdots, N\}$.  Given a  weighted undirected graph $G=(V,E,w)$ with $|E| = m, w \in \mathbb R^N_{+}$ and a parameter $\epsilon>0$, the idea is to privately choose a topology (edge set) and reweigh the weights of all edges in the new edge set. Specifically, we use $\widehat{m}$ (which will be specified in Algorithm \ref{alg1}) to denote the perturbed number of edges. Recall that we will sample from the distribution $\pi$ over ${[N]\choose \widehat{m}}$:

\begin{equation}\label{eq.distribution}
  \begin{aligned}
    \forall \widehat{E}\in {[N] \choose \widehat{m}}, \quad \pi(\widehat{E})\propto \prod_{e\in \widehat{E}}\exp(\epsilon w_e).
  \end{aligned}
\end{equation}

\noindent The algorithm outputs a synthetic graph $\widehat{G}$ with at most $\widehat{m}$ edges by first sampling a new edge set $\widehat{E}$ from $\pi$. It then uses the Laplace mechanism to output the perturbed weights of edges in $\widehat{E}$. Note that we use $\widehat{m}$ instead of $m$ to define the distribution to preserve privacy. The first step of our algorithm can be considered as using the exponential mechanism (with $\ell_1$-norm as its scoring function) over a non-convex set to select a new edge set. 
We use the topology sampler (Theorem~\ref{t.topology_sampler}) to perform this task. Our algorithm is formally described as Algorithm~\ref{alg1}.

\begin{algorithm}[h]
	\caption{{Private spectral approximation by topology sampler}}\label{alg1}
	\KwIn{A graph $G = (V,E,w)$ with $w\in \mathbb{R}^N_{+}$ and $|E|=m$, privacy budgets $\epsilon$, and a parameter $\beta\in (0,\frac{1}{2})$.}
	\KwOut{A synthetic graph $\widehat{G}$.}

    Let $\widehat{m}\leftarrow \min\{N, \lceil m + \texttt{Lap}(1/\epsilon) + \log(1/\beta)/\epsilon \rceil \}$ \;
    Let $\widehat{E} \leftarrow \texttt{TS}_{\epsilon}(\widehat{m}, w)$\;
    \For{$e\in \widehat{E}$}{
        Draw an independent Laplace noise $Z \sim \lap(1/\epsilon)$ \;
        $\widehat{w}_e  \leftarrow \max \{0, w_e + Z$\} \;
    }
    \For{$e'\in [N] \land e' \notin \widehat{E}$}{
        $\widehat{w}_{e'} \leftarrow 0$ \;
    }
    \Return{$\widehat{G} = \left(V,\widehat{E}, \{\widehat{w}_e\}_{e\in \widehat{E}}\right)$}.

\end{algorithm}

\begin{remark}
  Note that in Algorithm \ref{alg1}, we add a constant $\log(1/\beta)/\epsilon$ to the perturbed number of edges $m+\texttt{Lap}(1/\epsilon)$. It makes $\widehat{m}\geq m$ with probability at least $1-\beta$. Thus, applying the topology sampler on $\widehat{m}$ has the opportunity to output an edge set $\widehat{E}$ that covers $E$, which plays a key role in the utility of spectral approximation.
\end{remark}

\subsection{Proof of Privacy and Utility Guarantee of Algorithm~\ref{alg1}}

\noindent We first give the privacy guarantee of Algorithm \ref{alg1}. Note that Algorithm \ref{alg1} is simply a combination of our private sampler and the Laplace mechanism, which implies the following privacy guarantee:

\begin{theorem}\label{t.spectral_privacy}
  Fix any $\epsilon>0$, Algorithm \ref{alg1} preserves $(4\epsilon,0)$-differential privacy.
\end{theorem}

\begin{proof}
  We show that sampling a topology by $\texttt{TS}_{\epsilon}$ (in Theorem \ref{t.topology_sampler}) preserves $(2\epsilon, 0)$-differentially privacy. This directly comes from the privacy of the exponential mechanism. In particular, let $G=(V,E,w)$ and $G'=(V,E',w')$ be a pair of neighboring graphs that differ in one edge by at most $1$, note that given any $k \in \mathbb{N}$ and $\widehat{E}\subseteq [n]$ with $|\widehat{E}| = k$,
  \begin{equation*}
    \begin{aligned}
      \frac{\mathsf{Pr}[\texttt{TS}_\epsilon({k,w}) = \widehat{E}]}{\mathsf{Pr}[\texttt{TS}_\epsilon({k,w'}) = \widehat{E}]} &= \frac{\prod_{e\in \widehat{E}} \exp(\epsilon \cdot w_e)}{\prod_{e\in \widehat{E}}\exp(\epsilon \cdot w'_e)} \cdot \frac{\sum_{E'\in {[N]\choose k}}\prod_{e\in {E'}}\exp(\epsilon \cdot w'_e)}{\sum_{E'\in {[N]\choose k}}\prod_{e\in {E'}}\exp(\epsilon \cdot w_e)}\\
      & \leq \max_{i\in [N]} \exp(2\epsilon |w_i - w'_i|)\leq e^{2\epsilon}.
    \end{aligned}
  \end{equation*} 
By Lemma \ref{l.laplace} and the post-processing property of differential privacy (Lemma \ref{l.post_processing}), we see that outputting $\widehat{m}$, $\widehat{G}$ by
  $$\widehat{m} \leftarrow  \min\{N, \lceil m + \texttt{Lap}(1/\epsilon) + \log(1/\beta)/\epsilon \rceil \}$$
  and
  $$ \widehat{G} \leftarrow \left(V, \widehat{E}, \{\widehat{w}_e\}_{e\in \widehat{E}}\right)$$
  are both $(\epsilon,0)$-differentially private. Thus, by the  composition theorem (Lemma \ref{l.adaptive_composition}) of differential privacy, Algorithm \ref{alg1} is $(4\epsilon,0)$-differentially private completing the proof of Theorem~\ref{t.spectral_privacy}.
\end{proof}

\noindent Next, we give the utility guarantee of Algorithm \ref{alg1}.

\begin{theorem}\label{t.spectral}
  Fix any $\epsilon>0$. For any input graph $G = (V,E,w)$ with maximum \textbf{unweighted} degree $\Delta(G)\leq n-1$ and $m$ satisfies that $m = \Omega(n)$ and $m = o(N)$, with probability at least $1-2\beta - 1/n^c$ for any positive constant $c$, our scheme outputs a $\widehat{G}$ such that   
  $$\left\|L_G-L_{\widehat{G}}\right\|_2 = O\left(\frac{\Delta(G)\cdot \log^2(n)}{\epsilon} + \frac{\log(1/\beta)\log^2n}{n\epsilon^2}\right).$$
\end{theorem}

Theorem \ref{t.spectral} improves on prior work~\cite{blocki2012johnson, dwork2014analyze, upadhyay2021differentially} whenever $\Delta(G) = o(\sqrt{n})$, which is usually the case in practice. Furthermore, unlike previous algorithms, the algorithm that achieves the guarantee in Theorem~\ref{t.spectral} preserves $\epsilon$-differential privacy. If we want a uniform bound on all class of graphs, including the graphs with $\Omega(\sqrt{n})$ unweighted edges, we can consider the following algorithm:
\begin{enumerate}
  \item Let $\hat{\Delta} \leftarrow \Delta(G) + \texttt{Lap}(1/\epsilon)$;
  \item If $\hat{\Delta} > \sqrt{n}$, then run the Analyze Gauss algorithm~\cite{dwork2014analyze} on the edge-adjacency matrix of $G, E_G$;  otherwise, run Algorithm \ref{alg1}.
\end{enumerate}

It is easy to see that it is an $(\epsilon,\delta)$-differentially private algorithm because Analyze Gauss~\cite{dwork2014analyze} is $(\epsilon,\delta)$-differentially private. Now, using Theorem \ref{t.spectral_privacy} and Theorem \ref{t.spectral} along with the accuracy guarantee of Analyze Gauss~\cite{dwork2014analyze} in outputting the covariance matrix, for any fixed $(\epsilon,\delta)$, we also have that, for any undirected graph $G$ with maximum unweighted degree $\Delta(G)$, it outputs a synthetic graph $\widehat{G}$ satisfying 
$$\left\|L_G-L_{\widehat{G}}\right\|_2 = O \left(\min\left\{\frac{\Delta(G)\cdot\log^2(n)}{\epsilon}, \frac{\sqrt{n}\log(n)\log(1/\delta)}{\epsilon}\right\}\right).$$

Before we start the proof of \Cref{t.spectral}, we collect some key lemmas and give a high-level overview of how we use them in our proof. The proof of \Cref{t.spectral}
is given at the very end of this subsection. Recall that $$\widehat m:=\min\{N, \lceil m + \texttt{Lap}(1/\epsilon) + \log(1/\beta)/\epsilon \rceil \}.$$ 
Therefore, from the tail inequality of the Laplace random variable, for any $\beta$ such that $0<\beta<1/2$, it holds with probability at least $1-2\beta$ that
$$m\leq \widehat{m} \leq m + 2\log(1/\beta)/\epsilon.$$
Note that $\widehat{m}$ is a differentially private count of the number of edges.
In this case, there will be a subset $E'\subseteq [N]$ with size $\widehat{m}$ (where $\widehat{m}\leq N$) that covers $E$. Namely, if we replace $m$ by $\widehat{m}$, we can view $E'$ as a new edge set, with at most $\lfloor 2\log(1/\beta)/\epsilon \rfloor$ additional edges with zero weight.

Therefore, for the sake of simplicity, we first concentrate on the case where the topology sampler works on $m$ instead of the perturbed number of edges $\widehat{m}$. Before we get into the proof, we make some notations explicit. For any weighted graph $G=(V,E,w)$ with $w \in \mathbb{R}^N$ and a subset of possible edges $S\subseteq [N]$, we denote by $G|S$ the restriction of $G$ in terms of $S$. Let $w(G|S) \in \mathbb R^N$ be the corresponding weight vector. That is, 
\[
w(G|S)_i =\begin{cases}
    w_i & i\in E\cap S; \\
    0 & \text{otherwise.}
\end{cases}
\]

For any $0\leq k\leq N-|S|$, we also let $[\mathcal{G}|S]_{n,k}$ be the random graph model where we choose $k$ edges uniformly at random among $[N]\backslash S$. For example, $[\mathcal{G}|\varnothing]_{n,k} = \mathcal{G}_{n,k}$ for $0\leq k\leq N$. For two (weighted) undirected graphs $G_1 = (V,E_1,w_1)$ and $G_2 = (V,E_2,w_2)$ with $E_1\cap E_2 = \varnothing$, we write $G_1 + G_2 $ as an undirected graph with edge set $E_1\cup E_2$. We use $\Delta(G)$ to denote the maximum unweighted degree of $G$ and $E(G)$ to denote the edges of the graph $G$. 

The first stage of our proof is to show that the maximum degree in the synthetic graph $\widehat{G}$ is not too large. In particular, in Lemma \ref{l.3}, we first show that, if we add some edges outside $E$ randomly, then the maximum unweighted degree of the new graph has an upper bound related to $\mathcal{G}_{n,2m}$. Subsequently, in Lemma \ref{l.1}, we show that the maximum degree in $\mathcal{G}_{n,2m}$ can be bounded by basic concentration bounds and Theorem \ref{t.random}. In Lemma \ref{l.5}, we relate previous two lemmas with the distribution $\pi$ and conclude that the maximum degree of $\widehat{G}$ can be bounded by $O(\Delta(G)\log(n))$ with high probability. Next, the second step is to use Lemma \ref{l.4} to show that the error in every edge in $E\cup \widehat{E}$ is roughly $O(\log(n)/\epsilon)$. Combining these two facts, we use Lemma \ref{l.spectral_bilu}
 to complete the proof. 

\begin{lemma}\label{l.3}
  For any $0\leq k\leq m$ and $a>0$, 
  $$\mathsf{Pr}\left[\Delta([\mathcal{G}|E]_{n, m-k} + G)\geq a+\Delta(G)\right]\leq \mathsf{Pr}[\Delta(\mathcal{G}_{n,2m})\geq a].$$
\end{lemma}

\begin{proof}
  Recall that $m = o(N)$, thus $m-k\leq m = o(N - m)$. Suppose $\mathcal{G}_{n,2m}\cap E = Y$. It's easy to verify that for any $Y\subseteq E$ and the random graph model $\mathcal{G}_{n,2m}$,
  $$\Delta\left((\mathcal{G}_{n,2m}|[N]\backslash E) + G \right)\geq \alpha+\Delta(G) \quad \text{implies that}  \quad \Delta(\mathcal{G}_{n,2m})\geq \alpha.$$ 
  Note that the edge set of $\mathcal{G}_{n,2m}|[N]\backslash E$ and $G$ are disjoint. Thus,  
  $$\mathsf{Pr}\left[\Delta((\mathcal{G}_{n,2m}|[N]\backslash E) + G)\geq \alpha+\Delta(G) ~\Big| \mathcal{G}_{n,2m}\cap E = Y\right]\leq \mathsf{Pr}\left[\Delta(\mathcal{G}_{n,2m})\geq \alpha ~\Big|\mathcal{G}_{n,2m}\cap E = Y\right].$$

 We also note that $2m-|Y|\geq m-k$ holds for all $Y\subseteq E$ and $0\leq k\leq m$. Then by simple coupling argument, for any $Y\subseteq E$, 
$$\mathsf{Pr}[\Delta([\mathcal{G}|E]_{n,m-k} + G)\geq a+\Delta(G)] \leq  \mathsf{Pr}[\Delta([\mathcal{G}|E]_{n,2m-|Y|} + G)\geq a + \Delta(G)],$$
since $[\mathcal{G}|E]_{n,2m-|Y|}$ chooses more edges.
Thus, 
\begin{equation*}
  \begin{aligned}
    \mathsf{Pr}[\Delta([\mathcal{G}|E]_{n,m-k} &+ G) \geq a+\Delta(G)] = \sum_{Y\subseteq E} \mathsf{Pr}\left[\Delta([\mathcal{G}|E]_{n,m-k} + G)\geq a+\Delta(G)\right]\mathsf{Pr}[\mathcal{G}_{n,2m}\cap E = Y]\\
    &\leq \sum_{Y\subseteq E} \mathsf{Pr}[\Delta([\mathcal{G}|E]_{2m-|Y|} + G)\geq a + \Delta(G)]\mathsf{Pr}[\mathcal{G}_{n,2m}\cap E = Y]\\
    &=\sum_{Y\subseteq E} \mathsf{Pr}\left[\Delta((\mathcal{G}_{n,2m}\left| \right.[N]\backslash E)+G) \geq a + \Delta(G)~\Big|~ \mathcal{G}_{n,2m}\cap E = Y\right] \mathsf{Pr}[\mathcal{G}_{n,2m}\cap E = Y]\\
    &\leq \sum_{Y\subseteq E} \mathsf{Pr}\left[\Delta(\mathcal{G}_{n,2m})\geq \alpha ~\Big|~\mathcal{G}_{n,2m}\cap E = Y\right]\mathsf{Pr}[\mathcal{G}_{n,m}\cap E = Y]\\
    & = \mathsf{Pr}[\Delta(\mathcal{G}_{n,2m})\geq a].
  \end{aligned}
\end{equation*}
This completes the proof of Lemma~\ref{l.3}.
\end{proof}

\begin{lemma}\label{l.1}
  For any $0<\gamma \leq 1/2$, $$\mathsf{Pr}\left[\Delta(\mathcal{G}_{n,2m})\geq 2\left(1+{1\over\gamma}\right)\Delta(G) \right] \leq 3n\cdot \exp\left(-\Omega \left({1 \over \gamma}\right) \right).$$
\end{lemma}

\begin{proof}
  Let $p = {2m\over N}  =\Omega\left({1\over n}\right)$ and $a=2\left(1+{1\over\gamma} \right)\Delta(G)$. We first show that $\mathsf{Pr}[\Delta(\mathcal{G}_{n,p})\geq a]$ is negligible. Fix any vertex $u\in [n]$, write the unweighted degree of $u$ as $d(u)$. In $\mathcal{G}_{n,p}$, we have
  $$\mathbb{E}[d(u)] = (n-1)p = 2m\cdot \frac{n-1}{N} \leq 2\Delta(G),$$
  this is because that the maximum unweighted degree in $G$ is $\Delta(G)$. By the multiplicative Chernoff bound, we see that for any $u\in [n]$,
    \begin{align*}
      \mathsf{Pr}\left[d(u)\geq 2\left(1+{1\over\gamma}\right)\Delta(G) \right] &\leq 
      \mathsf{Pr}\left[d(u) \geq \left(1+ {1\over \gamma}\right)\cdot\mathbb{E}[d(u)] \right] 
      \\
      &\leq \exp\left(- \frac{\mathbb{E}[d(u)]/\gamma^2}{2+1/\gamma}\right) \\
      & \leq \exp\left(- \Omega(1/\gamma)\right)
    \end{align*}
  since $\mathbb{E}[d(u)] = (n-1)p = \Omega(1)$ and $1/\gamma \geq 2$.
  By the union bound over all vertices in $[n]$, we have that with probability at least $1-n\cdot \exp(-\Omega(1/\gamma))$, the maximum degree of $\mathcal{G}_{n,p}$ is less than $2\left(1+{1\over\gamma}\right)\Delta(G)$. Next, we use Lemma \ref{t.random} to convert the conclusion on $\mathcal{G}_{n,p}$ to $\mathcal{G}_{n,2m}$, in which the concentration bound is not applicable since the presence of each edge is not mutually independent.

  Note that for our choice of $m = o(N)$ and $m = \Omega(n)$, we have $p = {2m \over N} = o(1)$, ${1 \over Np} = {1 \over 2m} = o(1)$, and $\sqrt{p \over N} = {\sqrt{2m} \over N} = o(1)$. For any $d\in \mathbb{N}_+$, let $\mathcal{P}_{d}$ be the set of all undirected unweighted graphs whose maximum degree is at least $d$. Then, $\mathcal{P}_{d}$ is a monotone increasing graph property since $G\in \mathcal{P}_{d}$ implies $G+e\in \mathcal{P}_{d}$. Let $d' = O\left(\left(1+{1 \over \gamma}\right)\Delta(G)\right)$,
  and replace $\mathcal{P}$ in Lemma \ref{t.random} by $\mathcal{P}_{d'}$, we have that 
  $$ \mathsf{Pr}\left[\Delta(\mathcal{G}_{n,2m})\geq 2\left(1+{1\over\gamma}\right)\Delta(G) \right] \leq 3n\cdot \exp(-\Omega(1/\gamma)),$$
  which completes the proof of Lemma \ref{l.1}. 
 \end{proof}
 
  Finally, we are ready to relate previous lemmas to the distribution $\pi$. Combining Lemma \ref{l.3} and Lemma \ref{l.1}, we have the following lemma:
  \begin{lemma}\label{l.5}
    For any $0<\gamma \leq 1/2$, $$\mathsf{Pr}\left[\Delta(\widehat{G}) \geq 3\Delta(G)\left(1+{1\over\gamma}\right) \right]\leq 3n\cdot \exp(-\Omega(1/\gamma)).$$
  \end{lemma}
\begin{proof}
  Let $G' = (V,\widehat{E})$ be an unweighted graph with $\widehat{E}$ be its edge set, where $\widehat{E}$ is the edge set we sampled from distribution $\pi$. Note that $\Delta(\widehat{G}) = \Delta(G')$. Thus, 
  \begin{equation*}
    \begin{aligned}
    \mathsf{Pr}\left[\Delta(\widehat{G})\right. &\geq  \left.3\Delta(G)\left(1+{1\over\gamma}\right) \right] = \mathsf{Pr}\left[\Delta(G') \geq  3\Delta(G)\left(1+{1\over\gamma}\right) \right] \\
      & = \sum_{Y\subseteq E} \mathsf{Pr}\left[\Delta(G') \geq  3\Delta(G)\left(1+{1\over\gamma}\right) \left|\widehat{E}\cap E \right. = Y \right] \cdot \mathsf{Pr}\left[\widehat{E}\cap E = Y \right]\\
      & = \sum_{Y\subseteq E} \mathsf{Pr}\left[\Delta([\mathcal{G}|E]_{n,m-|Y|} + G|Y) \geq  3\Delta(G)\left(1+{1\over\gamma}\right) \right] \cdot \mathsf{Pr}\left[\widehat{E}\cap E = Y \right] \\
      & \leq \sum_{Y\subseteq E} \mathsf{Pr}\left[\Delta([\mathcal{G}|E]_{n,m-|Y|} + G) \geq  3\Delta(G)\left(1+{1\over\gamma}\right) \right] \cdot \mathsf{Pr}\left[\widehat{E}\cap E = Y \right]\\
      & = \sum_{Y\subseteq E} \mathsf{Pr}\left[\Delta([\mathcal{G}|E]_{n,m-|Y|} + G) \geq \Delta(G) + \left( 3\Delta(G)\left(1+{1\over\gamma}\right) - \Delta(G) \right) \right] \cdot \mathsf{Pr}\left[\widehat{E}\cap E = Y \right]\\
      & \leq  \sum_{Y\subseteq E} \mathsf{Pr}\left[\Delta(\mathcal{G}_{n,2m})\geq  3\Delta(G)\left(1+{1\over\gamma}\right) - \Delta(G) \right] \cdot \mathsf{Pr}\left[\widehat{E}\cap E = Y \right] \\
      &  \leq  \mathsf{Pr}\left[\Delta(\mathcal{G}_{n,2m})\geq 2\Delta(G)\left(1+{1\over\gamma}\right) \right]\\
      &\leq 3n\cdot \exp(-\Omega(1/\gamma)),
    \end{aligned}
  \end{equation*}
  completing the proof of Lemma~\ref{l.5}
\end{proof}

  \noindent Next, we give the following lemma that relates the weight of the (non-existing) edges in the input and output graph:
  \begin{lemma}\label{l.4}
    Let $0<\gamma\leq 1/2$ be a parameter. With probability at least $1-(n^4+2n^2) \exp(-1/\gamma)$, 
    $$|w_e - \widehat{w}_e|\leq {1 \over \gamma\epsilon} \quad \text{ for all }e\in E\cup \widehat{E}.$$
  \end{lemma}

\begin{proof}
  We can proof this lemma by the fact that distribution $\pi$ preserves the existence for edges with large weights.
  We consider three disjoint collections of edges in $E\cup \widehat{E}$:
  \paragraph{(1) $e\in \widehat{E}$ but $e\notin E$} In this case, by the tail inequality of Laplace distribution, we have that for any such $e$,
  $$\mathsf{Pr}\left[|w_e - \widehat{w}_e|> {1 \over \gamma\epsilon} \right] = \mathsf{Pr}_{Z\sim \texttt{Lap}(1/\epsilon)}\left[Z > {1 \over \gamma\epsilon} \right]\leq \exp\left(-{1\over \gamma} \right).$$ 
  Thus, by the union bound, we see that with probability at least $1-n^2 \exp(-1/\gamma)$, 
  $$\forall e\in \widehat{E} \backslash E, \qquad  |w_e - \widehat{w}_e|\leq {1 \over \gamma\epsilon}.$$
  
  \paragraph{(2) $e\in E$ but $e\notin \widehat{E}$} We claim that there exists some constant $c$ such that $w_e < {1 \over \gamma\epsilon}$ for all such $e$ with probability at least $1-n^c \exp(-1/\gamma)$. It's sufficient to show that with such probability, all edges in $E$ with $w_e\geq {1 \over \gamma\epsilon}$ will be preserved in $\widehat{E}$. If no edges in $E$ with $w_e\geq {1 \over \gamma\epsilon}$, the statement holds trivially. Let $e^*$ be any edge in $E$ with $w_{e^*}\geq {1 \over \gamma\epsilon}$. Then, for any $Y\subseteq (E\backslash e^*)$,
  \begin{equation*}
    \begin{aligned}
      \mathsf{Pr}\left[e^*\in \widehat{E}~\Big|~\left({E}\backslash \{e^*\}\right)\cap \widehat{E} = Y \right] &\propto \exp({\epsilon w_{e^*}}) \cdot \left(\sum_{B\in {[N]\backslash E}\atop |B| = m-|Y|} \prod_{e\in Y}\exp(\epsilon w_e)\prod_{e\in B} \exp(\epsilon\cdot 0)\right)\\
      & = \exp({\epsilon w_{e^*}}) \cdot \left(\sum_{B\in {[N]\backslash E}\atop |B| = m-|Y|} \prod_{e\in Y}\exp(\epsilon w_e)\right) \propto \exp({\epsilon w_{e^*}}) \cdot {N-m \choose m-|Y|}.
    \end{aligned}
  \end{equation*}
  On the other hand,
  \begin{equation*}
    \begin{aligned}
      \mathsf{Pr}\left[e^*\notin \widehat{E}~\Big|~\left({E}\backslash \{e^*\} \right)\cap \widehat{E} = Y\right] &\propto \sum_{B\in {[N]\backslash E}\atop |B| = m-|Y| + 1} \prod_{e\in Y}\exp(\epsilon w_e)\prod_{e\in B} \exp(\epsilon\cdot 0)\\
      & = \sum_{B\in {[N]\backslash E}\atop |B| = m-|Y|+1} \prod_{e\in Y}\exp(\epsilon w_e)
       \propto {N-m \choose m-|Y| + 1}.
    \end{aligned}
  \end{equation*}

\noindent Since $N-m-|Y| \leq N$, we therefore have, for any $Y\subseteq (E\backslash e^*)$, 
\begin{equation*}
  \begin{aligned}
    \mathsf{Pr}\left[e^*\in \widehat{E}~\Big|~({E}\backslash {e^*})\cap \widehat{E} = Y\right] &= \frac{\mathsf{Pr}\left[e^*\in \widehat{E}|({E}\backslash {e^*})\cap \widehat{E} = Y\right]}{\mathsf{Pr}\left[e^*\in \widehat{E}|({E}\backslash {e^*})\cap \widehat{E} = Y\right] + \mathsf{Pr}\left[e^*\notin \widehat{E}|({E}\backslash {e^*})\cap \widehat{E} = Y\right]}\\
    & = \frac{e^{\epsilon w_{e^*}}{N-m \choose m-|Y| }}{{N-m \choose m-|Y| + 1} + e^{\epsilon w_{e^*}}{N-m \choose m-|Y| }}
     = \frac{e^{\epsilon w_{e^*}}}{e^{\epsilon w_{e^*}} + \frac{N-m-|Y|}{|Y|+1}} \geq \frac{e^{\epsilon w_{e^*}}}{e^{\epsilon w_{e^*}} +N}.
  \end{aligned},
\end{equation*}

Since $w_{e^*} \geq {1 \over \gamma\epsilon}$ for some $0<\gamma \leq 1/2$, then 
\begin{equation*}
  \begin{aligned}
    \mathsf{Pr}\left[e^*\in \widehat{E}~\big|~({E}\backslash {e^*})\cap \widehat{E} = Y\right] \geq \frac{e^{\epsilon w_{e^*}}}{e^{\epsilon w_{e^*}} +N} \geq \frac{\exp(1/\gamma)}{\exp(1/\gamma) + n^2} \geq 1-n^2 \exp(-1/\gamma)
  \end{aligned}
\end{equation*}
holds for any $Y\subseteq E\backslash e^*$. Thus, using total probability theorem, 
\begin{equation*}
  \begin{aligned}
    \mathsf{Pr}\left[e^*\in \widehat{E}\right] = \sum_{Y\subseteq (E\backslash e^*)} \mathsf{Pr}\left[e^*\in \widehat{E}~\big|~({E}\backslash {e^*})\cap \widehat{E} = Y\right] \cdot \mathsf{Pr}[({E}\backslash {e^*})\cap \widehat{E} = Y] \geq 1-n^2 \exp(-1/\gamma).
  \end{aligned}
\end{equation*}

Using the union bound, we have that all edges in $E$ with $w_e \geq {1 \over \gamma\epsilon}$ will be preserved in $\widehat{E}$ with a  probability at least $1-n^4 \exp(-1/\gamma)$. Namely, we have shown that with such probability, all edges in $E\backslash \widehat{E}$ have weights $w_e$ less than ${1 \over \gamma\epsilon}$. That is,  
$$\forall e\in E\backslash \widehat{E}, \qquad \left|w_e - \widehat{w}_e\right| = |w_e - 0| \leq {1 \over \gamma\epsilon}.$$

\paragraph{(3) $e\in E\cap \widehat{E}$} In this case, again by the tail inequality of Laplace distribution, we have that with probability at least $1-n^2\exp(-1/\gamma)$,
 $$|w_e - \widehat{w}_e|\leq {1 \over \gamma\epsilon}$$
 holds for all such $e$ simultaneously. Thus, by a union bound, we have that in all three cases, i.e., $e\in E\cup \widehat{E}$, 
$$ |w_e - \widehat{w}_e| \leq {1 \over \gamma\epsilon} $$
holds with probability at least $1-(n^4+2n^2) \exp(-1/\gamma)$. This completes the proof of Lemma~\ref{l.4}.
\end{proof}
\begin{remark}
  An edge $e$ not in $E\cup \widehat{E}$ means that this edge neither appears in the input graph $G$ nor in the output graph $\widehat{G}$. Thus, $w_e =\widehat{w}_e = 0$.
\end{remark}
Next, we apply the following result of Bilu and Linial to convert the approximation on cuts to the approximation on the spectrum.

\begin{lemma}
[Lemma 3.3 in \cite{bilu2006lifts}]
\label{l.spectral_bilu}
  Let $A\in \mathbb{R}^{n\times n}$ be a symmetric matrix such that the $\ell_1$ norm of each row of $A$ is at most $\ell$ and all diagonal entries of $A$ has absolute value less than $O(\alpha (\log(\ell/\alpha)+1))$. If for any non-zero $u,v\in \{0,1\}^n$ with $x^\top y = 0$, it holds that 
   $$\frac{|x^\top A y|}{\|x\|_2\|y\|_2} \leq \alpha,$$
   then the spectral radius of $A$ is $\|A\|_2=O(\alpha (\log(\ell/\alpha)+1))$.
   \end{lemma}

\begin{proof}
[Proof of \Cref{t.spectral}]
Equipped with all these lemma, we now give a proof of \Cref{t.spectral}, let $L_{\widetilde{G}} = L_G - L_{\widehat{G}}$. Clearly $L_{\widetilde{G}}\in \mathbb{R}^{n\times n}$ is a symmetric matrix such that $L_{\widetilde{G}} \mathbf{1}_n = \mathbf{0}$. We then check the conditions in Lemma \ref{l.spectral_bilu} for $L_{\widetilde{G}}$. Let $\delta = (n^4 + 2n^2 + 3n) \exp(-\Omega(1/\gamma))$. Using Lemma \ref{l.4} and Lemma \ref{l.5}, with probability at least $1-\delta$, we have  
\begin{equation}\label{e.4.31}
  \begin{aligned}
    \max_{i\in [n]} |L_{\widetilde{G}}[i,i]| &= \max_{i\in [n]}|\Phi_{G}(\{i\}) - \Phi_{\widehat{G}}(\{i\})| 
    \leq {1 \over \gamma\epsilon}(d_{G}(i) + d_{\widehat{G}}(i)) \\
    &\leq \frac{1}{\gamma\epsilon}\left(\Delta(G) +  3\Delta(G)\left(1+{1\over\gamma}\right) \right).
  \end{aligned}
\end{equation}

On the other hand, for all $\mathbf{1}_S,\mathbf{1}_T\in \{0,1\}^n$ such that $S\cap T = \varnothing$:
\begin{equation}\label{e.4.32}
  \begin{aligned}
    \frac{|\mathbf{1}_S^\top L_{\widetilde{G}}\mathbf{1}_T|}{\|\mathbf{1}_S\|_2 \|\mathbf{1}_T\|_2} &= \frac{1}{{\|\mathbf{1}_S\|_2 \|\mathbf{1}_T\|_2}}\sum_{\{i,j\}\in E} w_{\{i,j\}} |\mathbf{1}_S[i]\mathbf{1}_T[i] + \mathbf{1}_S[j]\mathbf{1}_T[j] - \mathbf{1}_S[i]\mathbf{1}_T[j] - \mathbf{1}_S[j]\mathbf{1}_T[i]| \\
    & = \frac{|\Phi_{\widetilde{G}}(S,T)|}{{\|\mathbf{1}_S\|_2 \|\mathbf{1}_T\|_2}} = \frac{|\Phi_{{G}}(S,T) - \Phi_{\widehat{G}}(S,T)|}{\|\mathbf{1}_S\|_2 \|\mathbf{1}_T\|_2} \\
    &\leq \frac{1}{\|\mathbf{1}_S\|_2 \|\mathbf{1}_T\|_2}\cdot \frac{1}{\gamma\epsilon}\cdot (\Delta(G) + \Delta(\widehat{G}))\cdot \min \left\{|S|,|T|\right\} \\
    &= \frac{1}{\gamma\epsilon}\left(\Delta(G) +  3\Delta(G)\left(1+{1\over\gamma}\right) \right) \cdot \min\left\{\sqrt{\frac{|S|}{|T|}}, \sqrt{\frac{|T|}{|S|}}\right\} \\
    &\leq \frac{1}{\gamma\epsilon}\left(\Delta(G) +  3\Delta(G)\left(1+{1\over\gamma}\right) \right).
  \end{aligned}
\end{equation}
Note that $\mathbf{1}_S[i]\mathbf{1}_T[j]$ and $\mathbf{1}_S[j]\mathbf{1}_T[i]$ cannot both be $1$, where $a[i]$ denote the $i$-th coordinate of a vector $a$. Then, we just let the $\alpha$ in Lemma \ref{l.spectral_bilu} be  
$$\alpha = \frac{1}{\gamma\epsilon}\left(\Delta(G) +  3\Delta(G)\left(1+{1\over\gamma}\right) \right).$$

By Lemma \ref{l.4}, given any vertex $v\in V$, there are at most $(\Delta(G) + 3\Delta(G)(1+1/\gamma))$ edges differ in their weight by at most $\frac{1}{\gamma\epsilon}$, namely
$$\forall u\in V, \quad \sum_{v \in V, u\neq v} |L_{\widetilde{G}}[u,v]| \leq \alpha.$$

Combing Equation (\ref{e.4.31}) implies that $\|L_{\widetilde{G}}\|_{\infty} \leq 2\alpha$, and all diagonal entries of $A$ has absolute value less than $\alpha(\log (\ell/\alpha) +1)$. 
Also, Equation \ref{e.4.32} implies that for all non-zero vectors $x,y\in \{0,1\}^n$ and $x^\top y  =0$, it holds that
$$\frac{|x^\top L_{\widetilde{G}} y|}{\|x\|_2\|y\|_2} \leq \alpha.$$

Finally, we let $\gamma = 1/(c'\log(n))$ for some large enough constant $c'$ so that $\delta = 1-1/n^c$, where $c$ is any positive constant. We have that
$$\alpha(\log(\ell/\alpha) + 1) =\frac{\log 2 + 1}{\gamma\epsilon}\left(\Delta(G) + 3\left(1+{1\over \gamma} \right) \Delta(G)\right) =O\left(\frac{\Delta(G)\cdot \log^2(n)}{\epsilon}\right).$$

Now we consider the case where the topology sampler samples $\widehat{m}$ edges instead of $m$ edges. We can equivalently consider the input graph $G$ as adding $O(\log(1/\beta)/\epsilon)$ new edges with zero weight uniformly on each vertex. This increases the maximum degree by at most $O(\log(1/\beta)/(n\epsilon))$. Combining this fact with Lemma \ref{l.spectral_bilu} completes the proof of Theorem \ref{t.spectral}.
\end{proof}

\section{Applications of Estimation of Graph Spectrum}\label{s.application}
In this section, we show some straightforward applications of our new $\widetilde{O}(\Delta(G)/\epsilon)$ utility bound on spectral approximation. We focus on several important parameters related to random walks on graphs:  {\em hitting time}, {\em commute time}, and {\em cover time}.

\subsection{Problem Definition}\label{s.problem_definition}

Given a graph $G = (V,E,w)$ with non-negative weights $w \in \mathbb R_+^N$, consider the simple random walk on the vertices of $G$ with the transition probability from vertex $u$ to vertex $v$ (where $u,v\in V$):

$$
\mathsf{Pr}(u\rightarrow v) = 
\begin{cases}
    \frac{w_{\{u, v\}}}{d(u)}, &\text{ if }v\sim u; \\
    0, &\text{ otherwise }
\end{cases},
\quad \text{where} \quad d(u) = \sum_{v\sim u} w_{\{u,v\}}
$$
is the weighted degree of $u$ and $u\sim v$ denotes that there is an edge between $u$ and $v$. This is the most standard type of random walk on graph. Let $\{X_i\}_{i\geq 1}$ be the random sequence of this random walk. Two of the most fundamental parameters related to this random walk is the {\em hitting time} and {\em commute time}, defined as follows:

\begin{definition}[Hitting time]
\label{def.commute}
  Given $n$-vertex graph $G = (V,E,w)$ with non-negative weights, $w \in \mathbb R_{+}^N$, for any $u,v\in V$ and $u\neq v$, the hitting time $h_{u,v}$ is defined as 
  $$h_{u,v} := \mathbb{E}[\min\{t\geq 1| X_1 = u \land X_t = v\}].$$
\end{definition}

\begin{definition}[Commute time]
    For any $u,v\in V$ and $u\neq v$,  the commute time is defined as 
  $$C_{u,v} := h_{u,v}+h_{v,u}.$$
\end{definition}

\noindent Another important statistic related to the random walk on $G$ is the cover time:
\begin{definition}
  Given $n$-vertex graph $G = (V,E,w)$ with non-negative weights $w \in \mathbb R_{+}^N$, let $\tau_v$ be the first time at which every vertex in $G$ is visited by the random walk initiated at $v\in [n]$. Then the cover time of $G$ is 
  $$\tau(G) := \max_{v\in [n]} \mathbb{E}[\tau_v].$$
\end{definition}
\noindent Our goal is to differentially privately release a synthetic graph that approximates the hitting time/commute time (between any pair of vertices) and the cover time of some connected graph $G$.

Before introducing our approach, we begin by presenting some useful definitions and facts. For any matrix $A$, we write $A^\dagger$ as its {\em Moore-Penrose pseudoinverse}. Let $G$ be a connected graph with non-negative edge weights, then its Laplacian matrix has the spectral decomposition
$L_G = \sum_{i=2}^{n}\lambda_i u_iu^\top_i$ with $\lambda_i>0$ for $2\leq i\leq n$.
One can verify that 
$$L_G^\dagger = \sum_{i=2}^{n}\frac{1}{\lambda_i}u_i u^\top_i.$$ 
We note that $L_G L_G^\dagger = L^\dagger_G L_G = \sum_{i=1}^{n-1} u_i u_i^\top$, which is the identity operator on $\text{im}(L_G)$, where $\text{im}(L_G)$ is the image (column) space of $L_G$. Let $\mathbf{1}_n$ be the all one vector of dimension $n$, if $G$ is connected, then $\text{im}(L_G) = \mathbf{1}_n^\perp$, that is the subspace of $\mathbb R^n$ orthogonal to all one vector space.
If $A$ is a symmetric matrix, we write $A\succeq 0$ if all eigenvalues of $A$ are non-negative (or equivalently $A$ is positive semi-definite) and $A\succeq B$ if $A-B \succeq 0$. 
\begin{fact}
Let $A,B$ be two symmetric matrices, if $A\succeq B$, then $A^\dagger \preceq B^\dagger$.
\end{fact}
\noindent For any $u,v\in V$, we write $b_{u,v} = \mathbf e_u -\mathbf e_v$, where $\{\mathbf e_i\}_{i\in [n]}$ is the unit orthonormal basis of $\mathbb{R}^n$. 

One important parameter studied in spectral graph theory is the {\em effective resistance}. It describes the electrical property of the graph. 
\begin{definition}
Given a weighted graph $G$ with $n$ vertices, for any $u,v\in V$ and $u\neq v$, the effective resistance between $u$ and $v$ is defined by $R_{\mathsf{eff}}(u,v) = b_{u,v}^\top L_G^\dagger b_{u,v}$.
\label{defn.effective_resistance}
\end{definition}
 Recall that we write the {\em spectral gap} $\lambda(G)$ as the smallest non-zero eigenvalue of $L_G$. In this section, we only consider the case where the input graph is connected (otherwise the hitting/commute or cover time can be infinite). Note that for a connected graph $G$, $\lambda(G) = \lambda_2(G)$.

\subsection{Approximating the Spectrum of $L_G^\dagger$}

It has been widely known that hitting time and commute time are determined by the {\em all-pairs effective resistances}~\cite{chandra1989electrical}, while the cover time can be similarly approximated~\cite{matthews1988covering}. 
The effective resistance is given by a quadratic form of  $L_G^\dagger$. Thus, we use the following lemma to convert the spectral approximation on $L_G$ to the spectral approximation on $L_G^\dagger$:

\begin{lemma}\label{l.spectrum_to_eff}
    Let $G$ and $\widehat{G}$ be two connected graphs with $n$ vertices and non-negative edge weights. 
    Suppose there exists $u,\zeta >0$ such that $\widehat{G}$ satisfies $\|L_G - L_{\widehat{G}}\|_2\leq \zeta $ and the spectral gap of $L_G$ is bounded by $\lambda(G)\geq 1/u$. If $\zeta u <1$, then
    $$\left\|L_G^\dagger - L_{\widehat{G}}^\dagger\right\|_2 \leq \frac{\zeta u^2}{1-\zeta u}.$$
\end{lemma}
\begin{proof}
  Let $\mathbf{I}_{\text{im}(L_G)} = L_G L_G^{\dagger}  = L_G^{\dagger}L_G$ be the identity operator in $\text{im}(L_G)$. We note that $\mathbf{I}_{\text{im}(L_G)} = \mathbf{I}_{\text{im}(L_{\widehat{G}})} = \mathbf{1}_n^{\perp}$ since both $G$ and $\widehat{G}$ are connected. Since $\|L_G - L_{\widehat{G}}\|_2\leq \zeta$, for all $x\in \text{im}(L_G)$, we have 
  \begin{equation*}
    \begin{aligned}
      x^\top L_{\widehat{G}}x \leq x^\top L_G x + \zeta x^\top x &= x^\top (\zeta \mathbf{I} + L_G)x = x^\top (\zeta \mathbf{I}_{\text{im}(L_G)} + L_G)x
    \end{aligned}
  \end{equation*}
  While for those $x\notin \text{im}(L_G)$, $x^\top L_{\widehat{G}}x \leq x^\top L_G x + \zeta x^\top x$ holds trivially. Thus, 
  $$L_{\widehat{G}} \preceq L_G + \zeta \mathbf{I}_{\text{im}(G)}.$$
  This implies that 
  \begin{equation*}
    \begin{aligned}
      L_{\widehat{G}}^\dagger &\succeq (L_G + \zeta \mathbf{I}_{\text{im}(G)})^\dagger  = ((\mathbf{I}_{\text{im}(G)} + \zeta L_G^\dagger)L_G)^{\dagger} = L_G^\dagger (\mathbf{I}_{\text{im}(G)} + \zeta L_G^\dagger)^\dagger = L_G^\dagger \sum_{k=0}^\infty (-\zeta L_G^\dagger)^k \\
      &= L_G^\dagger + \sum_{k=1}^\infty (-1)^k \zeta^k (L_G^\dagger)^{k+1},
    \end{aligned}
  \end{equation*}
  where the third equality comes from Taylor's expansion. Our aim is to bound the series on the right-hand side. Note that for all unit vector $x\in \mathbb{R}^n$, 
  \begin{equation*}
    \begin{aligned}
      x^\top L_{\widehat{G}}^\dagger x &\geq x^\top L_{G}^\dagger  x + x^\top \left(\sum_{k=1}^\infty (-1)^k \zeta^k (L_G^\dagger)^{k+1}\right) x \\
      & \geq  x^\top L_{G}^\dagger x - \max_{z^\top z = 1}z^\top \left(\sum_{k=1}^\infty (-1)^k \zeta^k (L_G^\dagger)^{k+1}\right) z\\
      &\geq  x^\top L_{G}^\dagger x  - \left\|\sum_{k=1}^\infty (-1)^k \zeta^k (L_G^\dagger)^{k+1}\right\|_2 \\
      &\geq  x^\top L_{G}^\dagger x - \|L_G^\dagger\|_2^k \sum_{k=1}^\infty \zeta^k \|L_G^\dagger\|_2^{k} \\
      & \geq x^\top L_{G}^\dagger x - \frac{\zeta u^2}{1-\zeta u}.
    \end{aligned}
  \end{equation*}
  The last inequality comes from the fact that $\lambda(G) = 1/\|L_G^\dagger\|_2$. On the other hand, since  
  $$L_{\widehat{G}} \succeq L_G - \zeta \mathbf{I}_{\text{im}(G)},$$
  we have 
  $$x^\top L_{\widehat{G}}^\dagger x \leq x^\top L_{G}^\dagger x + \frac{\zeta u^2}{1-\zeta u}.$$
  Combining all of this, we see that $$\left\|L_G^\dagger - L_{\widehat{G}}^\dagger\right\|_2 \leq \frac{\zeta u^2}{1-\zeta u}$$
  completing the proof of Lemma~\ref{l.spectrum_to_eff}. 
\end{proof}

For a given weighted and simple graph $G$, our unified approach to privately approximating the all-pair hitting time, commute time and the cover time,  proceeds by reducing to approximating all pair effective resistances on $G$. This can be done by privately approximating the spectrum of $G$, which
we get for free from Algorithm \ref{alg1}. To put these approximations in context, we first analyze the sensitivity of one-pair effective resistance. 

\subsection{Sensitivity of One-pair Effective Resistance and the Laplace Mechanism}

Lemme \ref{l.spectrum_to_eff} already shows that if there is a graph $\widehat{G}$ approximates the spectrum of $G$, then the effective resistances should be close between $G$ and $\widehat{G}$, since $\|L_G^\dagger - L_{\widehat{G}}^\dagger\|_2$ is bounded. This is not only the key step in our reduction from computing commute and cover time to approximating the spectrum of $G$, but also implies an upper bound on the sensitivity of one-pair effective resistance, which we formalize in the follow theorem:

\begin{theorem}\label{t.sensitivity_eff}
    Fix any constant $\kappa>0$. For large enough $n$, there exists a pair of neighboring graphs $G = (V,{E},{w})$ and $G' = (V,{E}',{w}')$ on $n$ vertices with $\lambda(G) \geq 2+ \kappa$, such that the difference in the most sensitive one-pair effective resistance is of the order $\Theta(1/\lambda^2(G))$. That is,
    $$\max_{u,v\in V\atop u\neq v} \left|R_{\mathsf{eff}}(u,v) - R'_{\mathsf{eff}}(u,v)\right| = \Theta\left( {1 \over \lambda^2(G)}\right).$$
\end{theorem}
\begin{remark}
    We note that for any pair of neighboring graph $(G,G')$, as long as $\kappa$ is a constant and one of them, say $G$, satisfies that $\lambda(G) \geq 2+ \kappa$, then by Lemma \ref{l.sens_of_eigengap}, one can easily verify that $1/\lambda^2(G') = \Theta(1/\lambda^2(G))$. This guarantees the symmetry of Theorem \ref{t.sensitivity_eff}
.\end{remark}

\begin{proof}
    (Of \Cref{t.sensitivity_eff}.) We first show that if $\lambda(G) \geq 2+ \kappa$, then for any pair of neighboring graphs $G$ and $G'$, the sensitivity of one-pair effective resistance is upper bounded by $O(1/(\lambda^2(G)))$. By the definition of edge-level differential privacy, the spectral radius of $L_G - L_{G'}$ is bounded by $\|L_G - L_{G'}\|_2\leq 2$. Using Lemma \ref{l.spectrum_to_eff}, we have that, if $\lambda(G)\geq 2+\kappa$, then for any $u,v\in [n]$ and $u\neq v$, the difference in effective resistance between $u,v$ in $G$ and $G'$ is $$\left|R_{\mathsf{eff}}(u,v) - R'_{\mathsf{eff}}(u,v)\right| = O\left(\frac{1}{1-\frac{2}{2+\kappa}}\cdot\frac{1}{\lambda^2(G)}\right).$$
    
    Next, we show that, in the worst case, the sensitivity of one-pair effective resistance can be as large as $\Omega(1/\lambda^2(G))$. Consider a pair of neighboring graphs where $G = K_n$ is the unweighted complete graph on $n$ vertices, and $G'$ is produced by removing an arbitrary edge in $G$, say $\{u,v\}$. By series and parallel composition law of the electrical circuit, the effective resistance between $u,v$ in $G$ is $R_{\mathsf{eff}}(u,v) = \frac{2}{n}$. Let $R'_{\mathsf{eff}}(u,v)$ be the effective resistance between $u,v$ in $G'$. Then by parallel composition, we have 
$$1 + \frac{1}{R'_{\mathsf{eff}}(u,v)} = \frac{n}{2},$$
which implies that
$$|R_{\mathsf{eff}}(u,v) - R'_{\mathsf{eff}}(u,v)| = \frac{1}{\frac{n}{2} - 1} - \frac{2}{n} = \frac{2}{n\left(\frac{n}{2} - 1\right)} = \Theta\left(\frac{1}{n^2}\right).$$
On the other hand, for complete graph $G$, $\lambda(G) = \lambda_2(G) = n$, and thus $$|R_{\mathsf{eff}}(u,v) - R'_{\mathsf{eff}}(u,v)| =  \Omega\left(\frac{1}{\lambda^2(G)}\right).$$
This completes the proof of Theorem~\ref{t.sensitivity_eff}.
\end{proof}

\noindent Next, we discuss some implications of Theorem \ref{t.sensitivity_eff}. Given the sensitivity of effective resistances between pairs of vertices, a first attempt may be to use the Laplace mechanism to perturb the value according to its sensitivity. Then, by the advanced composition of differential privacy~\cite{dwork2010boosting} and the Laplace mechanism (Lemma \ref{l.laplace}), to output all-pairs effective resistances, using the Laplace mechanism suffers an error as large as $\Omega\left({n \over \epsilon}\cdot\frac{1}{\lambda^2(G)}\right)$ to achieve $(\epsilon,\delta)$-DP (where $\delta$ is negligible), since there are ${n \choose 2}$ pairs of vertices and the advanced composition of pure DP algorithms is no longer pure DP.

\subsection{A Unified Spectral Approach}\label{s.app}

For a given weighted and simple graph $G$, we show that the output of \Cref{alg1} is already good enough for privately approximating the all-pair hitting time, commute time, and cover time. Compared to the Laplace mechanism, our spectral approach reduces the error on approximating all-pairs effective resistances from $\Omega\left(\frac{1}{\lambda^2(G)}\cdot {n \over \epsilon}\right)$ to $\widetilde{O}\left(\frac{1}{\lambda^2(G)}\cdot {\Delta(G) \over \epsilon}\right)$ so that we have a better approximation on these graph parameters.
We first use the following lemma to relate commute time and effective resistance:

\begin{lemma}[Chandra et al.~\cite{chandra1989electrical}]\label{l.commute_time}
  Given a weighted graph $G = (V,E,w)$ with $\|w\|_1 = \sum_{e\in E}w_e$, for any pair of vertices $u,v\in V$:
  $$C_{u,v} = 2\|w\|_1\cdot R_{\mathsf{eff}}(u,v).$$
\end{lemma}
\noindent We can now use Matthews's bound \cite{matthews1988covering} to relate the commute time and cover time:
\begin{lemma}[Matthews~ \cite{matthews1988covering}]\label{l.covering_time}
  Given a graph $G = (V,E)$, the cover time of $G$ satisfies that
  $$\frac{\max_{u,v\in [n]}C_{u,v}}{2} \leq \tau(G)\leq \left(\max_{u,v\in [n]}C_{u,v}\right)(1+\log(n)).$$
\end{lemma}

We are now ready to present our result on the differentially private approximation of the commute time and cover time. Recall that we write $\lambda(G)$ to denote the spectral gap of $L_G$. Note that if $G$ is connected, then $\lambda(G) = \lambda_2(G)$, where $\lambda_2(G)$ is the second smallest eigenvalue of $L_G$.

\subsubsection{Private Approximation of Commute Time}
Our main result in this section is the following bound.
\begin{theorem}\label{t.app_in_commute}
Let $\epsilon>0$ be the privacy parameter. There is a polynomial time $(\epsilon,0)$-differentially private algorithm $\mathcal{A}$ which takes as input a connected graph $G=(V,E,w)$ with $w\in \mathbb{R}_{+}^{{n\choose 2}}$, $\lambda(G) = \Omega\left(\frac{\Delta(G)\log^2(n)}{\epsilon}\right)$, and $\frac{\Delta(G)\log^2(n)}{\epsilon \lambda(G)}<1/2$, and  outputs a vector $ \widehat{C} := \mathcal{A}(G) \in \mathbb{R}^N$ such that with high probability, for any $u,v\in V$, 
$$\left|\widehat{C}[\{u,v\}] - C_{u,v}(G) \right| = O\left(\frac{\|w\|_1}{\lambda^2(G)} \cdot \frac{\Delta(G)\log^2(n)}{\epsilon}\right).$$
\end{theorem}

\noindent Note that this approximation is non-trivial as long as 
$$O\left(\frac{1}{\lambda^2(G)} \cdot \frac{\Delta(G)\log^2(n)}{\epsilon}\right) < \max_{\{u,v\} \in {[n]\choose 2}}R_{\mathsf{eff}}(u,v),$$
which happens in sparse graphs with large weights, since the left-hand side decreases quadratically as the weights grow while the right-hand decreases linearly.

\begin{proof}
(Of Theorem \ref{t.app_in_commute}) Given privacy budget $\epsilon$, to construct the algorithm $\mathcal{A}$, for given graph $G$, we use Algorithm \ref{alg1} with privacy budget $\epsilon/8$ to output a synthetic graph $\widetilde{G}$. If $\widetilde{G}$ is not connected, we just overlay a complete graph $K_n$ with edge weight $1/n$ on $\widetilde{G}$, and let the new graph be $\widehat{G}$. Since $\left\|{1\over n}\cdot L_{K_n} \right\|_2 = O(1)$, then by the utility guarantee of Algorithm \ref{alg1} (Theorem \ref{t.spectral}), we have that with high probability,
  $$\left\|L_{\widehat{G}} - L_G\right\|_2 = O\left(\frac{\Delta(G)\log^2(n)}{\epsilon}\right).$$ 

Let $\widehat{W}  \leftarrow \|w\|_1 + \texttt{Lap}(2/\epsilon)$ be the perturbed value of the sum of edge  weights. Note that the sensitivity of the $\ell_1$ norm of graphs is $1$.
\noindent Then, we set $\widehat{C}[\{u,v\}] = \widehat{W}\cdot b^\top_{u,v} L_{\widehat{G}}^\dagger  b^\top_{u,v}$ as the approximation of the commute time between vertices $u$ and $v$. By Theorem \ref{t.spectral_privacy} and the post-processing property of differential privacy (Lemma \ref{l.post_processing}), we see that $\mathcal{A}$ preserves $(\epsilon,0)$-pure differential privacy. The utility part directly follows from Lemma \ref{l.commute_time} and Lemma \ref{l.spectrum_to_eff}, combined with the assumption of $$\frac{1}{\lambda(G)} \cdot \frac{\Delta(G)\log^2(n)}{\epsilon }<1/2.$$
This completes the proof of Theorem \ref{t.app_in_commute}.
\end{proof}

\subsubsection{Private Approximation on Cover Time}
\noindent Next, we give the result on approximating the cover time with multiplicative error.

\begin{theorem}\label{t.app_in_covering_time}
  Fix an $\epsilon>0$. For any connected graph $G = (V,E,w)$ with $w\in \mathbb{R}_{+}^{n\choose 2}$, let $\mathcal{A}'$ be the algorithm that outputs the maximum value in the output of $\mathcal{A}$ (in Theorem \ref{t.app_in_commute}). Then $\mathcal{A}'$ is $\epsilon$-DP. Moreover, for any connected graph $G\in \mathbb{R}_{+}^{n\choose 2} $ with $\lambda(G) = \Omega\left(\frac{\Delta(G)\log^2(n)}{\epsilon}\right)$ and $\frac{\Delta(G)\log^2(n)}{\epsilon \lambda(G)}<1/2$,  $\mathcal{A}'$ outputs a $\hat{\tau}\in \mathbb{R}$ such that with high probability, 
  $${\tau(G) \over \gamma} - \xi \leq \hat{\tau} \leq \gamma\tau(G) + \xi,$$
  where, $\gamma = 1+\log(n)$ and $\xi = O\left(\frac{\|w\|_1}{\lambda^2(G)} \cdot \frac{\Delta(G)\log^2(n)}{\epsilon}\right)$.
  \end{theorem}

\begin{proof}
   Let $\widehat{C}$ be the output of $\mathcal{A}$.For any $u\in [n]$ we set $C_{u,u} = \widehat{C}[\{u,u\}] = 0$. Then, it is easy to verify that
  \begin{equation*}
    \begin{aligned}
      \tau(G) &\leq \left(\max_{u,v\in [n]} C_{u,v}\right)(1+\log(n)) \\
      &\leq \left(\max_{u,v\in [n]} \widehat{C}[\{u,v\}] + \xi \right)(1+\log(n))
      \\
      &\leq (\hat{\tau}+\xi)(1+\log(n)),
    \end{aligned}
  \end{equation*}
  where the first inequality comes from Lemma \ref{l.covering_time} and the second inequality comes from Theorem \ref{t.app_in_commute}. By rearranging the terms, we have $\hat{\tau}\geq \tau(G)/\gamma - \xi$. The other direction is similar. This completes the proof.
\end{proof}

\subsubsection{Private Approximation on Hitting Time in $\ell_\infty$- and $\ell_2$-norm}
In this section, we show that an approximation on the graph spectrum also implies an approximation on the hitting time of any pair of vertices. The hitting time plays a crucial role in analyzing the structure of graphs~\cite{tetali1991random, lovasz1993random, von2014hitting}. 
 
Consider a weighted connected graph $G = (V,E,w)$. For any fixed $t\in V$, denote by $\mathbf{h}_{*,t}\in \mathbb{R}^n$ the vector of hitting times from any vertex to $t$ such that $ \mathbf{h}_{*,t}[i] = h_{i,t}$ ($h_{t,t}=0$). By Tetali's formula~\cite{tetali1991random}, for any $u,v\in [n]$ and $u\neq v$, we can write the hitting time from $u$ to $v$ as:

$$h_{u,v} = \|w\|_1\cdot R_{\mathsf{eff}}(u,v) + \sum_{i\in [n]}\frac{d(i)}{2}(R_{\mathsf{eff}}(v,i) - R_{\mathsf{eff}}(i,u)).$$

By Lemma \ref{l.spectrum_to_eff} and the utility guarantee of Algorithm \ref{alg1} (Theorem \ref{t.spectral_privacy}), using Tetali's formula over the private output of Algorithm \ref{alg1} (after overlaying a complete graph of weight $1/n$ if necessary as in the proof of Theorem \ref{t.app_in_commute}), directly implies that there is an algorithm such that with high probability,  for every $t\in [n]$, it outputs a $\hat{{\mathbf{h}}}_{*,t}\in \mathbb{R}^n$ where

$$\|\hat{\mathbf{h}}_{*,t} - {\mathbf{h}}_{*,t}\|_{\infty} = O\left(\frac{\Delta(G)(\|w\|_1 + nW_{\max}\Delta(G))}{\lambda^2(G)} \cdot \frac{\log^2(n)}{\epsilon}\right).$$

If we further consider the approximation in terms of the $\ell_2$ norm, we need to pay an extra $\sqrt{n}$ factor since $\|x\|_2 \leq \sqrt{n}\|x\|_{\infty}$ for any $x\in \mathbb{R}^n$. To obtain a more accurate approximation, we propose another approach which is based on  $L_G^\dagger$ instead of all-pairs effective resistances.
\begin{theorem}\label{t.app_in_hitting_time}
    Fix a $\epsilon>0$. There is an efficient and $(\epsilon,0)$-differentially private algorithm $\mathcal{A}''$ which takes as input a connected graph $G$ with weighted vector $w \in \mathbb{R}_{+}^{n\choose 2}$, $\lambda(G) = \Omega\left(\frac{\Delta(G)\log^2(n)}{\epsilon}\right)$ and $\frac{\Delta(G)\log^2(n)}{\epsilon \lambda(G)}<{1 \over 2}$, and outputs $n$ vectors $\{\hat{\mathbf{h}}_{*,t}\}_{t\in [n]}\in (\mathbb{R}^n)^n$ such that with high probability, for any $t\in [n]$,
    $$\|\hat{\mathbf{h}}_{*,t} - {\mathbf{h}}_{*,t}\|_2 = O\left(\frac{\sqrt{n}\Delta(G)\left(\|w\|_1 + \frac{\sqrt{n}\log(n)}{\epsilon}\right)}{\lambda^2(G)}  \cdot\frac{\log^2(n)}{\epsilon}\right).$$
\end{theorem}

\begin{proof}
    We first relate the hitting time and graph Laplacian by standard techniques. Fix a vertex $t\in V$. Let $h_{u,t}$ be the hitting time from $u$ to $t$, for any $u\neq t$ and $u\in V$. It is easy to verify that 
    \begin{equation}\label{e.hitting}
        h_{u,t} = 1+ \frac{w_{u,v}}{d(u)}\sum_{v\sim u} h_{v,t}
    \end{equation}
    Now, consider $ \mathbf{h}_{*,t}\in \mathbb{R}^n$ where $ \mathbf{h}_{*,t}[i] = h_{i,t}$. By re-arranging items in \Cref{e.hitting}, we see that $ \mathbf{h}_{*,t}$ satisfies
    \begin{equation}\label{e.hitting2}
        L_G \mathbf{h}_{*,t} = \mathbf{d}^{t,G},
    \end{equation}
    where 
    $$
    \mathbf{d}^{t,G} = 
    \begin{cases}
        d(u) & u \neq t \\
        d(t) - 2\|w\|_1, &\text{otherwise,}
    \end{cases}
    $$
    and $w$ is the weight vector of the graph $G$.
    Since $\mathbf{d}^{t,G} \in \mathbf{1}_n^\perp$, then $\mathbf{h}_{*,t}$ from \Cref{e.hitting2} has solution $\mathbf{h}_{*,t} = \{L_G^\dagger\mathbf{d}^{t,G} + c\mathbf{1}_n, c\in \mathbb{R}\}$. Thus, we design the algorithm to approximate hitting times as follows:
\begin{enumerate}
  \item Given graph $G$, run Algorithm \ref{alg1} on $G$ with privacy budget $\epsilon/8$.
  \item Sample $n$ i.i.d Laplace noise $z_i\sim \texttt{Lap}(4/\epsilon)$.
  \item Publish $\hat{\mathbf{d}}\in \mathbb{R}^n$ as the perturbed degree sequence, where $\hat{\mathbf{d}}[i] = \mathbf d(i)+z_i$. And for any $t\in [n]$, let $\hat{\mathbf{d}}^t = \mathbf{d}^{t,G} + \mathbf{z}^t$, where 
   $$
  \mathbf{z}^t = \left\{
  \begin{aligned}
    &z_u, &\text{if }u\neq t;\\
    &z_t - \sum_{i\in [n]} z_i &\text{otherwise}.
  \end{aligned}
  \right.
  $$
  \item For any $t\in [n]$, compute $\hat{c}_t = -\left(L_{\widehat{G}}^\dagger  \hat{\mathbf{d}}^t\right)[t]$ (which is the $t$-th element of $L_{\widehat{G}}^\dagger  \hat{\mathbf{d}}^t$), and output $$\hat{\mathbf{h}}_{*,t} =L_{\widehat{G}}^\dagger  \hat{\mathbf{d}}^t + \hat{c}_t \mathbf{1}_n.$$
\end{enumerate}

For the privacy part of this algorithm, consider any fixed $t$, $\hat{\mathbf{d}}^t$ can be obtained by $\hat{\mathbf{d}}$. Also, changing the weight of one edge by at most $1$ affects the degrees of two vertices. Thus, by Theorem \ref{t.spectral_privacy} and Lemma \ref{l.laplace}, this algorithm preserves $(\epsilon,0)$-differential privacy. 

We note that $\hat{\mathbf{d}}$ might have negative entries. However, we only use $L_{\widehat{G}}^\dagger  \hat{\mathbf{d}}^t$ to approximate $L_{{G}}^\dagger  {\mathbf{d}}^{t,G}$, instead of considering initiating a random walk on the synthetic graph $\widehat{G}$.

For the utility part, let $t$ be a fixed vertex. Clearly we have that $\mathbf{h}_{*,t} = L_G^\dagger \mathbf{d}^{t,G} + c_t\mathbf{1}_n$, where $c_t$ satisfies that $c_t + (L_G^\dagger \mathbf{d}^{t,G})[t] = 0$. Here, $(L_G^\dagger \mathbf{d}^{t,G})[t]$ is the $t$-th element in $ (L_G^\dagger \mathbf{d}^{t,G})[t]$. By the tail bound of linear combinations of Laplace noise, with probability at least $1-{1\over \text{poly}(n)}$, $$\sum_{i\in [n]} z_i = O\left({\sqrt{n}\log(n) \over \epsilon} \right),$$ 
and for any $i\in [n]$, $|z_i|\leq {\log(n) \over \epsilon}$. Therefore, with a high probability,
$$\|\mathbf{z}^t\|_2 \leq \|\mathbf{z}\|_2 + O\left({\sqrt{n}\over \epsilon}\right),$$
which is at most $O(\sqrt{n}\log(n)/\epsilon)$. Thus,

\begin{equation*}
  \begin{aligned}
    \left\|L_G^\dagger \mathbf{d}^{t,G} - L_{\widehat{G}}^\dagger  \hat{\mathbf{d}}^t\right\|_2 &= \left\|L_G^\dagger \mathbf{d}^{t,G} - L_{\widehat{G}}^\dagger  \mathbf{d}^{t,G} + L_{\widehat{G}}^\dagger\mathbf{z}^t\right\|_2\\
    &\leq  \left\|L_G^\dagger - L_{\widehat{G}}^\dagger \right\|_2 \cdot \|\mathbf{d}^{t,G}\|_2 + \left\|L_{\widehat{G}}^\dagger\right\|_2\cdot \|\mathbf{z}^t\|_2\\
    &\leq  \left\|L_G^\dagger - L_{\widehat{G}}^\dagger \right\|_2 \cdot \|\mathbf{d}^{t,G}\|_2 + \left(\left\|L_{{G}}^\dagger\right\|_2 + \left\|L_{\widehat{G}}^\dagger - L_{{G}}^\dagger \right\|_2\right)\cdot \|\mathbf{z}^t\|_2\\
    &= O\left(\frac{\|w\|_1}{\lambda^2(G)} \cdot \frac{\Delta(G)\log^2(n)}{\epsilon}\right) + O\left({\sqrt{n} \over \lambda^2(G)}\cdot \frac{\Delta(G)\log ^3 n}{\epsilon^2} \right).
  \end{aligned} 
\end{equation*}
Therefore, for any fixed $t\in [n]$,
\begin{equation*}
\begin{aligned}
  \|\mathbf{h}_{*,t} - \hat{\mathbf{h}}_{*,t} \|_2 &= \left\|L_G^\dagger \mathbf{d}^{t,G} + c_t \mathbf{1}_n - L_{\widehat{G}}^\dagger  \hat{\mathbf{d}}^t - \hat{c}_t \mathbf{1}_n \right\|_2\\
  &\leq \left\|L_G^\dagger \mathbf{d}^{t,G} - L_{\widehat{G}}^\dagger  \hat{\mathbf{d}}^t\right\|_2 + (c_t - \hat{c}_t)\cdot \|\mathbf{1}_n\|_2\\
  &\leq (\sqrt{n}+1)\cdot \left\|L_G^\dagger \mathbf{d}^{t,G} - L_{\widehat{G}}^\dagger  \hat{\mathbf{d}}^t\right\|_2\\
  &= O\left(\frac{\|w\|_1}{\lambda^2(G)} \cdot \frac{\sqrt{n}\Delta(G)\log^2(n)}{\epsilon} + {n \over \lambda^2(G)}\cdot \frac{\Delta(G)\log ^3 n}{\epsilon^2}  \right),
\end{aligned}
\end{equation*}
which completes the proof.
 \end{proof}

\paragraph{Comparison of the Laplace mechanism and our spectral approach}

In Section \ref{s.app}, we have showed that our algorithm implies an approximation on all-pairs effective resistances with error at most $\widetilde{O}\left(\Delta(G)/\epsilon\cdot\frac{1}{\lambda^2(G)}\right)$, and we achieve this utility with pure DP. Equivalently, our spectral approximation matches the sensitivity of one-pair effective resistance up to a factor of $\Delta(G)$. Compared to the Laplace mechanism, we are able to release a synthetic graph and reduce the error by a factor of $\Delta(G)/n$, while also achieving pure differential privacy. However, using the Laplace mechanism requires a much-relaxed assumption on the range of $\lambda(G)$, which means that the Laplace mechanism could serve as an alternative to our spectral approach in the case where $\lambda(G)$ fails to satisfy the assumption stated in Theorem \ref{t.app_in_commute}, Theorem \ref{t.app_in_covering_time} and Theorem \ref{t.app_in_hitting_time}.

\section{Differentially Private Cut Approximation}
As mentioned in the introduction, the algorithm for spectral approximation also guarantees cut approximation but has an error that scales with maximum degree. In this section, we describe our algorithm for private cut approximation with optimal error rate. 

\subsection{The Algorithm}
Let $N = {n\choose 2}$, recall that given an undirected weighted graph $G = (V,E, w)$, our goal is to output a synthetic graph $\widehat{G} = (V, \widehat E, \widehat w)$ with $\widehat w \in \mathbb{R}^N$ efficiently and privately, while $\widehat{G}$ approximates $G$ in terms of the sizes of all $(S,T)$-cuts up to an additive error $\widetilde{O}(\sqrt{n|E|}/\epsilon)$. One of the subroutines to achieve this is to use the topology sampler $\texttt{TS}_{\epsilon}$ (Section \ref{s.sampler}), which samples from a distribution over some non-convex set with the constraint on sparsity.

To achieve the optimal error on cut approximation, we combine our topology sampler and private mirror descent (\cite{eliavs2020differentially}). We use private mirror descent as a black box:

\begin{theorem}[\cite{eliavs2020differentially}]\label{t.pmd}
  Fix $0<\epsilon,\delta<1/2$, there is a $(\epsilon,\delta)$-differentially private algorithm such that for any weighted graph $G$ with $n$ vertices, it output a synthetic graph $\widetilde{G} = (V,\widetilde E, \widetilde w)$ with $\widetilde w \in \mathbb{R}_{+}^N$ such that with probability at least $1-\beta$, for all disjoint $S,T\subseteq [n]$,
  $$|\Phi_{G}(S,T) - \Phi_{\widehat{G}}(S,T)| = O\left( \sqrt{\frac{n\|w\|_1}{\beta^2\epsilon}}\log^2\left(\frac{n}{\delta}\right)\right).$$
  Further, this algorithm runs in time $\widetilde{O}(n^7)$.
\end{theorem}

For a weighted $n$-vertex, $m$-edge graph $G=(V,E,w)$, private mirror descent outputs a synthetic graph privately with an $\widetilde{O}\left(\sqrt{n\|w\|_1 \over \epsilon} \right)$ purely additive error (in the worst case) on approximating the sizes of all $(S,T)$-cuts. To entirely remove the dependency on the edge weights, we observe that if an edge $e\in [N]$ is not chosen by the topology sampler, then with high probability, its weight satisfies that $w_e = O(\log(n)/\epsilon)$. Therefore, the idea is to use a two-step strategy:
\begin{enumerate}
  \item Use $\texttt{TS}_\epsilon$ ensured by Theorem~\ref{t.topology_sampler} to sample and publish a new topology $\widehat{E}$;
  \item Run private mirror descent promised by Theorem~\ref{t.pmd} on the sub-graph containing edges in $E\backslash \widehat{E}$.
\end{enumerate}

From the union bound, with high probability, the sub-graph given to the private mirror descent has maximum weight at most $O(\log(n)/\epsilon)$. For the edges that are chosen for the new edge set, we use the tail bound of independent Laplace random variables to bound the error in cut approximation, which also does not rely on the edge weights. The algorithm is summarized in Algorithm \ref{alg2}. In Algorithm \ref{alg2}, we write $\beta$ as a parameter to be determined.

\begin{algorithm}[h]
	\caption{{Private cut approximation by topology sampler}}\label{alg2}
	\KwIn{A graph $G = (V,E,w)$ with $m$ edges (where $w\in \mathbb{R}^N_{+}$), privacy budgets $\epsilon, \delta$.}
	\KwOut{A synthetic graph $\widehat{G}$.}

    Let $\widehat{m}\leftarrow \min\{N, \lceil m + \texttt{Lap}(1/\epsilon) + \log(1/\beta)/\epsilon \rceil \}$ \;
    Let $\widehat{E} \leftarrow \texttt{TS}_{\epsilon}(\widehat{m},w)$\;
    \For{$e\in \widehat{E}$}{
        Draw an independent Laplace noise $Z \sim \lap(1/\epsilon)$ \;
        $\widehat{w}_e  \leftarrow  w_e + Z$ \;
    }
    Let $\widehat{G}_1 \leftarrow \left(V,\widehat{E}, \{\widehat{w}_e\}_{e\in \widehat{E}}\right)$ \;
    Run private mirror descent (\Cref{t.pmd}) on the sub-graph $\left(V, E\backslash \widehat{E}, \{w_e\}_{e\in  E\backslash \widehat{E}}\right)$ with parameter $(\epsilon, \delta)$ \;
    Let $\widehat{G}_2$ be the output of private mirror descent \;
    \Return{$\widehat{G} = \widehat{G}_1 + \widehat{G}_2$}.

\end{algorithm}

 Next, we introduce the privacy and utility guarantee of Algorithm \ref{alg2}.
\begin{theorem}
  For any $0<\epsilon,\delta<1/2$, Algorithm \ref{alg2} preserves $(5 \epsilon, \delta)$-approximate differential privacy. 
\end{theorem}
\begin{proof}
  For any fixed $0<\epsilon,\delta<1/2$, note the topology sampler $\texttt{TS}_{\epsilon}$ preserves $(2\epsilon,0)$-pure differential privacy (\Cref{t.spectral_privacy}), and the private mirror descent is $(\epsilon,\delta)$-approximate differentially private (\Cref{t.pmd}). Then, this theorem directly follows from the privacy guarantee of the Laplace mechanism (Lemma \ref{l.laplace}) and the adaptive composition of differential privacy (Lemma \ref{l.adaptive_composition}).
\end{proof}

Before presenting the utility of Algorithm \ref{alg2}, we introduce the tail inequality of independent combinations of Laplace noise:

\begin{lemma}
  [Gupta et al.~\cite{gupta2012iterative}]
  \label{l.tail_of_laplace}Given $k\in \mathbb{N}_{>0}$ and $b>0$. Let $\{Z_i\}_{i\in [k]}$ be i.i.d random variables such that $Z_i\in \texttt{Lap}(b)$, then for any $q_1,\cdots q_k\in [0,1]$, 
  $$\mathsf{Pr}\left[\sum_{i\in [k]} q_iZ_i > \alpha\right] \leq \left\{
      \begin{aligned}
          &\exp\left(-\frac{\alpha^2}{6kb^2}\right), \quad &\alpha \leq kb, \\
          &\exp\left(-\frac{\alpha}{6kb}\right), &\alpha > kb.
      \end{aligned}
      \right.$$
  \end{lemma}

\begin{theorem}\label{t.optimal_error}
  Let $\epsilon \in \left(\frac{1}{n},\frac{1}{2} \right)$ and $0<\delta<\frac{1}{2}$ be the privacy parameters. For any input weighted graph $G$ with $n$ vertices and $m$ edges $(m\geq n)$, with probability at least $1-o(1)$, Algorithm \ref{alg2} outputs a synthetic graph such that for all disjoint $S,T\subseteq [n]$,
  $$\left|\Phi_{G}(S,T) - \Phi_{\widehat{G}}(S,T)\right| = O\left( {\frac{\sqrt{nm}}{\epsilon}}\log^3\left(\frac{n}{\delta}\right)\right).$$
\end{theorem}

\begin{proof}
  Let $G = (V,E,\{w_e\}_{e\in E})$ with $|E| = m$. Suppose the topology sampler $\texttt{TS}_\epsilon$ chooses $\widehat{E} \subseteq [N]$ with size $|\widehat{E}| = \widehat{m}$. Let $G_1 = (V, \widehat{E} , \{w_e\}_{e\in \widehat{E}})$ and $G_2 = (V, E\backslash \widehat{E}, \{w_e\}_{e\in  E\backslash \widehat{E}})$, where $G_2$ is exactly the input of private mirror descent in Algorithm \ref{alg2}. Clearly $G = G_1 + G_2$ since $w_e = 0$ for all $e\notin E$. 
  By Lemma \ref{l.4}, we see that with high probability, 
  $$\max_{e\in E\backslash \widehat{E}} w_e = O\left(\log \left({n\over \epsilon}\right)\right).$$
Thus, the sum of weights of edges in $G_2$ is $O\left({|E|\log(n) \over \epsilon} \right)$. From the utility guarantee in Theorem \ref{t.pmd}, we have that with probability at least $1-\beta - 1/n^c$, for any disjoint $S,T \subseteq [n]$, 
  \begin{equation*}
    \begin{aligned}
      \left|\Phi_{G_2}(S,T) - \Phi_{\widehat{G}_2}(S,T)\right| = O\left( {\frac{\sqrt{n|E|\log(n)}}{\beta^2\epsilon}}\log^2\left(\frac{n}{\delta}\right)\right).
    \end{aligned}
  \end{equation*}

\noindent Next, we consider the difference between $G_1$ and $\widehat{G}_1$. We prove the following lemma:

\begin{lemma}\label{l.topology_is_known}
Fix any $\epsilon\geq \frac{1}{n}$. With probability at least $1-1/n^c - 2\beta$, for any disjoint $S,T \subseteq [n]$, 
\begin{equation*}
  \begin{aligned}
    \left|\Phi_{G_1}(S,T) - \Phi_{\widehat{G}_1}(S,T)\right| = O\left( \frac{\sqrt{\log(1/\beta)}}{\epsilon}{\sqrt{n|E|\log^2 (n)}}\right).
  \end{aligned}
\end{equation*}
\end{lemma}
\begin{proof}
For any pair of disjoint sets of vertices $S$ and $T$, let $\widehat{E}(S,T)$ be the collection of edges in the given graph that crosses $S$ and $T$. Note that given $G_1$ and $\widehat{G}_1$ with mutual edge set $\widehat{E}$, and any disjoint $S,T\subseteq [n]$, the difference in the cut size can be written as $$\texttt{err}(S,T) = \sum_{e \in \widehat{E}(S,T)} Z_e,$$ 
where $\{Z_e\}_{e\in |\widehat{E}|}$ are i.i.d Laplace noise from $\texttt{Lap}(1/\epsilon)$. 
  
  Let $\alpha(\widehat{E}) = \min\{|\widehat{E}|\cdot \lceil\log(n)\rceil,{n\choose 2}\}$. Since $|\widehat{E}|\geq |\widehat{E}(S,T)|$, one can further find a set of $\alpha(\widehat{E})$ edges in a complete graph that includes $\widehat{E}(S,T)$, let that set be $\tau(\widehat{E})$. Then, the error can be re-written as 
  $$\texttt{err}(S,T) = \sum_{e\in \tau(\widehat{E})} q_e Z_e \quad \text{where} \quad q_e = \begin{cases}
    1 & e\in \tau(\widehat{E})\cap \widehat{E}(S,T); \\
    0 & \text{otherwise.}
  \end{cases}$$ 
  
  Note that with probability at least $1-\beta$, $n \leq m \leq \widehat{m}$. Conditioned on $\widehat{m} \geq n$, we have  
  \begin{equation*}
    \begin{aligned}
      \frac{\sqrt{n|\tau(\widehat{E})|\log(n)}}{2\epsilon}  &\leq \frac{1}{\epsilon}\min \left\{\sqrt{n|\widehat{E}|}\log(n), O(n^{1.5}\sqrt{\log(n)}) \right\} \leq \frac{1}{\epsilon}{|\tau(\widehat{E})|}.
    \end{aligned}
  \end{equation*} 
  Thus, by Lemma \ref{l.tail_of_laplace}, we have 
  \begin{equation}
    \mathsf{Pr}\left[\sum_{e\in \tau(S)} q_eZ_e > \frac{\sqrt{n|\tau(\widehat{E})|\log(n)}}{2\epsilon} \right] \leq \exp\left(-\frac{n|\tau(\widehat{E})|\log(n)/\epsilon^2}{6|\tau(\widehat{E})|/\epsilon^2}\right) = \exp\left(-\frac{n\log(n)}{6}\right).
  \end{equation}
  Since there are at most $2^{2n}$ different $(S,T)$-cuts, then we have 
  \begin{equation*}
    \begin{aligned}
      &\mathsf{Pr}\left[\exists \text{ disjoint }S,T\in 2^{[n]} \texttt{ s.t.} |S|\leq |T| \land \texttt{err}(S,T) > \frac{\log(n)}{2\epsilon}\sqrt{n|\widehat{E}|}\right] \\
      &\quad \leq \mathsf{Pr}\left[\exists \text{ disjoint }S,T\in 2^{[n]} \texttt{ s.t.} |S|\leq |T| \land \texttt{err}(S,T) > \frac{1}{2\epsilon}\sqrt{n\log(n)|\tau(\widehat{E})|}\right]\\
      &\quad \leq \exp\left(-\Theta(\log(n))\right) \\
      & \quad = O\left(\frac{1}{\texttt{poly}(n)}\right).
    \end{aligned}
  \end{equation*}
  Thus, we have that with probability at least $1-O(1/n^c)$, the error is at most $O\left(\frac{\sqrt{n\widehat{m}\log ^2 (n)}}{\epsilon}\right)$. On the other hand, since 
  $$\widehat{m} \leq m + \texttt{Lap}(1/\epsilon)+ \log(1/\beta)/\epsilon,$$
  then with probability at least $1-\beta$, $\widehat{m}\leq m + 2\log(1/\beta)/\epsilon = O(m\log (1/\beta))$ (with $\epsilon\geq 1/n$), which completes the proof of Lemma \ref{l.topology_is_known}.
\end{proof}

Let the output graph be $\widehat{G} = (V,\widehat{E}, \{\tilde{w}_e\}_{{e}\in \widehat{E}})$. Let $E^* = E \cup \widehat{E}$ be the edges that have non-zero weight in either $G$ or $\widehat{G}$. With the utility guarantee on $\widehat{G}_1$ and $\widehat{G}_2$ formed in Algorithm~\ref{alg2}, for any disjoint $S,T\subseteq [n]$, we have 
\begin{equation*}
  \begin{aligned}
    |\Phi_{G}(S,T) - \Phi_{\widehat{G}}(S,T)| &= ~\Big| \sum_{e\in E^*(S,T)}  w_e - \sum_{e\in E^*(S,T)} \tilde{w}_e ~\Big|\\
    & =  |\Phi_{G_1}(S,T) +  \Phi_{G_2}(S,T) - (\Phi_{\widehat{G}_1}(S,T)  + \Phi_{\widehat{G}_2}(S,T))|\\
    & \leq |\Phi_{G_1}(S,T) - \Phi_{\widehat{G}_1}(S,T)| + |\Phi_{G_2}(S,T) - \Phi_{\widehat{G}_2}(S,T)|\\
    & =  O\left( {\frac{\sqrt{n|E|\log(n)}}{\beta^2\epsilon}}\log^2\left(\frac{n}{\delta}\right)\right).
  \end{aligned}
\end{equation*}
Setting $\beta = {1 \over \log^{0.25}n}$ completes the proof of Theorem \ref{t.optimal_error}.
\end{proof}

\subsection{Lower Bound}\label{s.lower_bound}

Here, we discuss the lower bound on differentially private approximation on $(S,T)$-cuts. It has been shown in \cite{eliavs2020differentially} that, for an unweighted graph $G=(V,E)$, there is an $\Omega\left(\sqrt{n\|w\|_0/\epsilon}\right)$ lower bound on this problem, and this lower bound matches the upper bound for unweighted graphs. However, the $\Omega(\sqrt{n\|w\|_0/\epsilon})$ lower bound is not applicable when the edge weights are sufficiently large. 
For example, there is an upper bound on private cut approximation with an $\widetilde{O}(n^{1.5}/\epsilon)$ error, which is independent of the edge weights, and makes the lower bound in \cite{eliavs2020differentially} valid when $\|w\|_1 = \omega(n^2/\epsilon)$. Thus, we give another lower bound on the additive error that tightly matches the upper bound when $\|w\|_1$ is large (compared to the number of edges $\|w\|_0$). In particular, we have the following theorem:

\begin{theorem}
\label{t.lowerbound_cut}
  Fix $\epsilon>0$ and $0<\delta<1$. If there is a $(\epsilon,\delta)$-differentially private algorithm such that for any $G=(V,E,w)$ with $w\in \mathbb{R}_{+}^{n\choose 2}$, it outputs a synthetic graph that with probability at least $1-\beta$, it approximates the sizes of all $(S,T)$-cuts of $G$ with a purely additive error $\alpha$, then 
  $$\alpha = \Omega\left(\frac{\sqrt{n|E|}}{\epsilon}\cdot \left(1-\frac{e-1}{e^\epsilon - 1} \cdot \frac{9\delta}{\beta}\right)\right).$$
\end{theorem}

\begin{proof}
  The proof follows a similar route as  \cite[Theorem 1.2]{eliavs2020differentially}, in which given any undirected graph $G$ with $n$ vertices and $m$ edges (of weight $1$), it has been known that there is no $(1,\delta)$-DP algorithm $\mathcal{M}$ whose error is with probability $\beta$ is below $o\left(\sqrt{mn}(1-\frac{9\delta}{\beta})\right)$. For the sake of contradiction, suppose there is a $(\epsilon,\delta)$-DP algorithm $\mathcal{M}'$ whose purely additive error is $o\left(\frac{\sqrt{mn}}{\epsilon}\cdot (1-c)\right)$, where $c = \frac{e-1}{e^\epsilon - 1} \cdot \frac{9\delta}{\beta}$. Then, if we reweigh $G$ to $G' = \frac{1}{\epsilon} G$, by the group privacy, we see that $\epsilon \cdot \mathcal{M'}\left(G'\right)$ is $(1, \frac{e-1}{e^\epsilon - 1}\delta)$-DP. However, it has error 
  $$\epsilon \cdot o\left(\frac{\sqrt{n|E|}}{\epsilon}\cdot \left(1-c\right)\right) = o\left({\sqrt{n|E|}}\cdot \left(1-c\right)\right),$$
  which leads to a contradiction. 
\end{proof}

\section{Discussions}

In this paper, we give new error bounds on differentially private spectral and cut approximation. Notably, there is still a gap in our result for spectral approximation. In particular, for a weighted graph with maximum unweighted degree $\Delta$ and average unweighted degree $\Delta_{\mathsf{avg}}$, the best known lower bound says any private algorithm for spectral approximation has an additive error at least $\Omega(\sqrt{\Delta_{\mathsf{avg}}})$~\cite{eliavs2020differentially}, while our algorithm achieves error within $\widetilde{O}(\Delta)$. How to close this gap remains open. 

For cut approximation, we show that for a weighted graph, the asymptotic bound $\widetilde{O}(\sqrt{\|w\|_0 n}/\epsilon)$ correctly characterizes the error rate, and we also give a polynomial time algorithm for achieving this error bound. However, we believe this is not the end of this line of work. Our $\widetilde{O}(\sqrt{\|w\|_0 n}/\epsilon)$ bound is \textit{purely} additive. \cite{blocki2012johnson} shows that if we allow a constant multiplicative error, then for weighted graphs with weights polynomial in $n$, the additive error can be reduced to $\widetilde{O}(n)$ by sampling from all its all possible weighted sparsifiers with the exponential mechanism. Currently, it remains an open question whether there exists an algorithm capable of efficiently achieving this error bound. (see also the discussion in \cite{eliavs2020differentially}).

While most results on differential privacy emphasize on privacy-utility trade-off, from a practical perspective, another major concern is the resource requirement, i.e., the time and space complexity. 
The time complexity of our algorithm (and the previously best-known algorithm~\cite{eliavs2020differentially}) on cut approximation is $\widetilde{O}(n^7)$ while the space complexity is $O(n^2)$. This might be prohibitive in many real-world applications -- for example, social networks often contain millions of users resulting in $n$ to be in the order of $10^6$. From the perspective of space complexity, common large-scale networks are often sparse with vertex degrees negligible compared to the scale of the users, making $O(n^2)$ prohibitive in practice. To address these practical challenges, in a companion work, we propose private algorithms for cut and spectral approximation in almost linear time and linear space, and we achieve a utility bound which, while sub-optimal for dense unweighted graphs, but holds with high probability instead of only in expectation. Note that, this matches the time and space assumption of a non-private algorithm for outputting a synthetic graph. In profit-driven organizations, analyses of sensitive graph data are usually implemented in a non-private way due to the absence of efficient algorithms, which is clearly a damage to user privacy. Consequently, from a practical point of view, we believe that introducing a linear time algorithm has the potential to bring positive social impacts to the privacy community.

\newpage
\begin{table}
\centering
\begin{tabular}{c|c}
\hline 
    Symbol & Meaning \\ \hline 
    $\mathbb R$ & set of real numbers \\ \hline
    $\mathbb R_+$ & set of non-negative real numbers \\ \hline
    $[n]$ & the set $\{1, 2,\cdots, n\}$ \\ \hline
    ${[N]\choose k}$ & set of all subset of $[N]$ of size $k$ \\ \hline 
    $G$ & graph   \\ \hline
    $n$ & number of vertices \\ \hline
    $N$ & ${n \choose 2}$ \\ \hline
    $m$ & number of edges \\ \hline
    $w$ & weight vector of the graph in $\mathbb R_+^N$ \\ \hline
    $\Delta(G)$ & maximum unweighted degree \\ \hline
    $\maxwt$ & maximum weight on an edge \\ \hline
    $L_G$ & Laplacian of a graph $G$ \\ \hline
    $\norm{v}_p$ & $\ell_p$ norm of a vector $v$ \\ \hline
    $\norm{A}$ & spectral norm of a matrix $A$ \\ \hline
    $\alpha$ & additive error on cut approximation \\ \hline
    $\mathsf{Pr}_\otimes$ & product distribution \\ \hline
    $\mathcal G_{n,m}$ & set of $(n,m)$ random graph with $m$ edges \\ \hline
    $\mathcal G_{n,p}$ & Erdos-Renyi random graph with $p \in (0,1)$ \\ \hline
    ${[k] \choose \ell}$ & all subsets of $[k]$ with cardinality $\ell$ \\ \hline
    $G \vert S$ & restriction of $G$ in terms of $S \subseteq E$ \\ \hline
    $[\mathcal G\vert E]_{n,k}$ & random graph with $k$ edges u.a.r. from $[N]\backslash E$ \\ \hline
    $G_1 + G_2$ & graph with edge sets $E_1 \cup E_2$ \\ \hline
    $d(u)$ & weighted degree of $u \in V$ \\ \hline
    $h_{u,v}$ & hitting time  for $u,v \in V$ \\ \hline
    $\tau(G)$ & cover time \\ \hline
    $C_{u,v}$ & commute time \\ \hline
    $R_{\mathsf{eff}}(u,v)$ & effective resistance between $u$ and $v$ \\ \hline 
\end{tabular}
\caption{Notations used in this paper}
\label{tab:notation}
\end{table}

\clearpage

\bibliographystyle{alpha}
\bibliography{privacy}
\newpage

\end{document}